\documentclass[12pt]{article}
\usepackage{graphicx}
\usepackage{natbib} 
\usepackage{url} 
\usepackage{bbm}

\usepackage{enumerate}
\usepackage{color}
\usepackage{float}
\usepackage{amsmath,amsthm,amssymb,amsfonts}
\newcommand{\blind}{0}

\newtheorem{thm}{Theorem}

\newtheorem{lemma}{Lemma}

\begin{document}

\def\spacingset#1{\renewcommand{\baselinestretch}%
{#1}\small\normalsize} \spacingset{1}


\if0\blind
{
  \title{\bf  An Effective Multivariate Normality Test via  Hessians of Empirical Cumulant Generating Functions}
 \author{Kwun Chuen Gary Chan
	\\
	Department of Biostatistics, University of Washington,\\
	\\
	Hok Kan Ling\\
	Department of Mathematics and Statistics, Queen's University\\
	\\
	Chuan-Fa Tang\\
	Mathematical Sciences, The University of Texas at Dallas\\
	\\
	Sheung Chi Phillip Yam\\
	Department of Statistics, The Chinese University of Hong Kong
}
\date{}
  \maketitle
} \fi

\if1\blind
{
  \bigskip
  \bigskip
  \bigskip
  \begin{center}
    {\LARGE\bf An Effective Multivariate Normality Test via  Hessians of Empirical Cumulant Generating Functions}
\end{center}
  \medskip
} \fi

\bigskip
\begin{abstract}
In this article, we propose a new class of consistent tests for $p$-variate normality. These tests are based on the characterization of the standard multivariate normal distribution, that the Hessian of the corresponding cumulant generating function is identical to the $p\times p$ identity matrix and the idea of decomposing the information from the joint distribution into the dependence copula and all marginal distributions. Under the null hypothesis of multivariate normality, our proposed test statistic is independent of the unknown mean vector and covariance matrix so that the distribution-free critical value of the test can be obtained by Monte Carlo simulation. We also derive the asymptotic null distribution of proposed test statistic and establish the consistency of the test against different fixed alternatives. Last but not least, a comprehensive and extensive Monte Carlo study also illustrates that our test is a superb yet computationally convenient competitor to many well-known existing test statistics. 
\end{abstract}

\noindent%

\vfill

\newpage
\spacingset{2} 
\section{Introduction}\label{sec:intro}
The normal distribution is  certainly one of the most important distributions in statistics that underlie many statistical procedures. Therefore, validation of the normality assumption behind the data is of fundamental importance and interest; to this end, there are a significant number of tests for both univariate and multivariate normality in the literature. For instance, \cite{mecklin2004appraisal} surveyed dozens of common multivariate normality tests and divided most of the tests into four categories: (i) procedures based on graphical plots and correlation coefficients; (ii) goodness-of-fit tests; (iii) test based on measures of skewness and kurtosis; (iv) and tests based on the empirical characteristic functions. Furthermore, there is still an ongoing research devoted in developing new normality tests. For example, \cite{ebner2020tests} reviewed the recent developments in tests for multivariate normality with an emphasis on asymptotic properties of several classes of weighted $L^2$-statistics, where these statistics are mostly based on empirical moment generating functions or empirical characteristic functions. 

%


In this article, we introduce a novel test for multivariate normality based on a system of second order partial derivatives of the cumulant generating function and decomposing a joint distribution into the dependence copula and all marginal distributions. 
Let $X$ be a random vector on $\mathbb{R}^p$. Denote $\mathcal{N}_p$ to be the class of all non-degenerate $p$-variate normal distributions. Suppose that we observe a random sample of $X_1,\ldots,X_n$ having the same joint distribution as that of $X$. We consider the traditional and conventional problem of testing the null hypothesis:
\begin{equation*}
	H_0: \text{The law of $X$ belongs to } \mathcal{N}_p,
\end{equation*}
against the general alternative hypothesis:
\begin{equation*}
	H_1: \text{The law of $X$ does not belong to } \mathcal{N}_p.
\end{equation*}
Let $\overline{X}_n := n^{-1}\sum^n_{i=1}X_i$ be the sample mean and $S_n := n^{-1}\sum^n_{i=1}(X_i- \overline{X}_n)(X_i-\overline{X}_n)^\top$ be the (biased) sample covariance matrix. Let $S^{-1/2}_n$ denote the unique symmetric square root of $S^{-1}_n$. For the rest of this article, we assume that $n \geq p + 1$. 
Due to the absolute continuity of the multivariate normal distribution, $S^{-1}_n$ exists with probability one under $H_0$, see for example \cite{eaton1973non}. Therefore, under the alternative hypothesis, if $S_n$ is not invertible, we can simply reject the null hypothesis. Thus, from now on we assume $S^{-1}_n$ exists.
It is known that, under the null hypothesis, the distribution of the so-called scaled residuals $Z_{n,i} := S^{-1/2}_n(X_i - \overline{X}_n)$, $i=1,\ldots,n$, does not depend on the mean vector $\mu$ and covariance matrix $\Sigma$ of $X$; also see \cite{szkutnik2021comment} for details. In other words, the distribution of the scaled residuals when  $X$ is from any non-degenerate multivariate normal distribution is the same as that when $X \sim N_p(0, \textbf{I}_p)$. Thus, for test statistics that only involve the scaled residuals $Z_{n,i}$'s, under the null hypothesis, it suffices to consider the case $X\sim N_p(0, \textbf{I}_p)$, where $\textbf{I}_p$ denotes the $p\times p$ identity matrix.
Denote $M(t) := \mathbb{E}(e^{t^\top X})$ to be the moment generating function of $X$. The cumulant generating function of $X$ is also defined by $\Lambda(t) := \log M(t)$, for $t$ in the proximity of origin. Our proposed test relies on the following key observation: if $X \sim N_p(0, \textbf{I}_p)$, then
\begin{equation}\label{eq:normal_cgf_Hess}
	H_\Lambda (t) = \textbf{I}_p,\quad t \in \mathbb{R}^p,
\end{equation}
where $H_f$ denotes the Hessian matrix of a function $f$; clearly,  the converse is also true, namely if the Hessian of the cumulant generating function of a random vector is identically equal to $\textbf{I}_p$, then the random vector has a standard multivariate normal distribution.

In general, a random vector can deviate from the multivariate normality from two causes: (i) at least one marginal distribution of a component random variable is non-normal; or (ii) the copula structure of the random vector is not the Gaussian one; see for example \cite{nelsen2007introduction}.
 In fact, practitioners often look at univariate normal Q-Q plots, P-P plots (\cite{gan1990goodness}) and bivariate plots of the marginals taken any two at a time, by performing univariate tests on each of the marginal distribution as well as performing tests based on dimension reduction (e.g., a test of the squared radii $Z_{n,i}^T Z_{n,i}$'s); see Chapter 9 in \cite{thode2002testing}. These create multiple testing issues, and using Bonferroni rule to combine them as an adjustment can be conservative.
Motivated by the fact that a multivariate distribution can be determined by specifying its corresponding marginal distributions and dependence structure, for any  $t = (t_1,\ldots,t_p)^\top \in \mathbb{R}^p$, define the function $M^*(\cdot)$ via the decomposition
\begin{equation*}
	M(t) \equiv \left( \prod^p_{j=1} M_j(t_j) \right) M^*(t),
\end{equation*}
where $M_{j}(\cdot)$ is the moment generating function of the $j$th component of $X$. 
Further, if $M$ is twice differentiable, then
\begin{equation}\label{eq:H_log_M_star}
	H_{\log M^*}(t) \equiv H_{\log M}(t) - \sum^p_{j=1} H_{\log M_j}(t_j).
\end{equation}
Denote $\Lambda_j := \log M_j$ and $\Lambda^* := \log M^*$. We rewrite (\ref{eq:H_log_M_star}) as 
\begin{equation}\label{eq:H_star_lambda_D}
	H_{\Lambda^*}(t) \equiv H_{\Lambda}(t) - D(t),
\end{equation}
where $D(t)$ is a $p \times p $ diagonal matrix with elements $H_{\Lambda_1}(t_1),\ldots,H_{\Lambda_p}(t_p)$ in order.
Under the null hypothesis $H_0$, by (\ref{eq:normal_cgf_Hess}),
\begin{equation}\label{eq:decomposition}
	H_{\Lambda^*}(t) = \textbf{0}_{p\times p} \text{ and } D(t) = \textbf{I}_p, \text{ for all } t \in \mathbb{R}^p,
\end{equation}
where $\textbf{0}_{p\times p}$ is the $p \times p$ zero matrix. Our proposed test is based on an empirical version of the two identities of (\ref{eq:decomposition}).

Recall that under $H_0$, the scaled residuals $Z_{n,i}$'s are independent of the unknown $\mu$ and $\Sigma$, and they will resemble a random sample from $N_p(0, \textbf{I}_p)$. Let $M^{(n)}(t) := \frac{1}{n}\sum^n_{i=1} e^{t^\top Z_{n,i}}$, $	\triangledown M^{(n)}(t) := \frac{1}{n}\sum^n_{i=1} Z_i e^{t^\top Z_{n,i}}$ and $	H_{M^{(n)}}(t) := \frac{1}{n}\sum^n_{i=1} Z_{n,i} Z_{n,i}^\top e^{t^\top Z_{n,i}}$ be the empirical version of the unknown theoretical moment generating function, its gradient and Hessian based on the scaled residuals, respectively.
An empirical version $H_{\Lambda^{(n)}}$ of $H_{\Lambda}$ using the scaled residuals is then given by
\begin{align}\label{eq:H_Lambda_n}
	H_{\Lambda^{(n)}}(t) := \frac{M^{(n)}(t) H_{M^{(n)}}(t) - \triangledown M^{(n)}(t)(\triangledown M^{(n)}(t))^\top}{M^{(n)}(t)^2},
\end{align}
where $\Lambda^{(n)} := \log M^{(n)}$. The corresponding empirical version of $H_{\Lambda^*_n}$ is 
\begin{equation*}
	H_{(\Lambda^{(n)})^*}(t) = H_{\Lambda^{(n)}}(t) - D^{(n)}(t),
\end{equation*}
where $D^{(n)}(t)$ is the $p \times p$ diagonal matrix with elements $H_{\Lambda^{(n)},11}(s_1),\ldots, H_{\Lambda^{(n)},pp}(s_p)$. Here, $s_j = (0,\ldots,0,t_j,0,\ldots,0)^\top$, where $t_j$ is in the $j$th position of this $p$-vector. 
In view of (\ref{eq:decomposition}), a large value of the $L^2$-statistics $n\int \cdots \int_{\mathbb{R}^p} \sum_{i=1}^p \sum^p_{j > i} H^2_{(\Lambda^{(n)})^*, ij}(t) dt$ or $n\int \cdots \int_{\mathbb{R}^p} \sum^p_{i=1} (D^{(n)}_{ii}(t) -1)^2 dt$ is an evidence against $H_0$.
Since these two integrals do not generally admit a closed form expression because of the denominator $M^{(n)}(t)$ in (\ref{eq:H_Lambda_n}), we consider their discretized approximations and define the following statistics:
\begin{align}\label{eq:Hn_Dn_t}
	H^{(n)}_{N} = n\sum^N_{l=1} \sum^p_{i=1}\sum^p_{j>i} H^2_{(\Lambda^{(n)})^*, ij}(t_l) \quad \text{ and } \quad 
	D^{(n)}_{N} = n\sum^N_{l=1}\sum^p_{i=1} (D^{(n)}_{ii}(t_l)-1)^2,
\end{align}
where $\{t_1,\ldots,t_N\}$ is a collection of vectors in $\mathbb{R}^p$. For example, we may choose them randomly in the ball $\{t \in \mathbb{R}^p : \|t \| \leq R \}$ for some $R > 0$ and $\|\cdot\|$ is the usual Euclidean norm. Our simulation studies show that $R = 3$ will work well in most scenarios. To see why a moderate value (like $3$) can be a good choice, we can write (\ref{eq:H_Lambda_n}) as
\begin{equation}\label{eq:alternative_weight}
	H_{\Lambda^{(n)}}(t) = \sum^n_{i=1} Z_{n,i} Z_{n,i}^\top W_{n,i}(t) - \left(\sum^n_{i=1}Z_{n,i} W_{n,i}\right) \left(\sum^n_{i=1}Z_{n,i} W_{n,i}\right)^\top,
\end{equation}
where $W_{n,i}(t) := \frac{e^{t^\top Z_{n,i}}}{\sum^n_{j=1} e^{t^\top Z_{n,j}}}$ and $H_{\Lambda^{(n)}}(t)$ can be interpreted as a weighted estimate of the covariance matrix. Under $H_0$,
 $Z_{n,i}$ behaves like a random vector from $N(0,\textbf{I}_p)$ under $H_0$ so that all the components of $Z_{n,i}$ will be around $-3$ to $3$ most of the time. A larger value of $\|t\|$ tend to put more weights to extreme values of $Z_{n,i}$'s in (\ref{eq:alternative_weight}). Thus, a moderate value of $R$ can avoid $H^{(n)}_N$ and $D^{(n)}_N$ depend heavily on only a few extreme values of $Z_{n,i}$'s. On the other hand, if $R$ is  small, then the weights are more even and less information in the empirical cumulant generating function is used, which may result in a less powerful test (recall that the population version identities in (\ref{eq:decomposition}) hold for all values of $t \in \mathbb{R}^p$ under $H_0$).
 
We can interpret $H^{(n)}_{N}$ as capturing the overall deviation from multivariate normality dependence structure while $D^{(n)}_{N}$ focuses on the deviation from marginal univariate normality. Note that $D^{(n)}_N$ is still a function of the scaled residuals, so that it does not depend on the mean and covariance matrix of the normal distribution under $H_0$.
The computation of $H^{(n)}_N$ and $D^{(n)}_N$ is straightforward as we essentially only have to compute the scaled residuals and the empirical Hessian of the cumulant generating function given in (\ref{eq:H_Lambda_n}) at different values of $t$.
Since the magnitudes of $H^{(n)}_N$ and $D^{(n)}_N$ are different, to combine evaluative assessment  $H^{(n)}_N$ and $D^{(n)}_N$, we define our proposed test statistic to be
\begin{equation}
T^{(n)}_N = \max\bigg\{ \frac{H^{(n)}_{N} - \mathbb{E}^S( H^{(n)}_{N})}{SD^S( H^{(n)}_{N}) },   
	\frac{D^{(n)}_{N} - \mathbb{E}^S( D^{(n)}_{N})}{SD^S( D^{(n)}_{N}) }\bigg\},
\end{equation}
such that we reject $H_0$ for large values of $T^{(n)}_N$. Here, the superscript ``$S$" refers to the calculations based on $S$ simulations: $\mathbb{E}^S(H^{(n)}_N) := \frac{1}{S}\sum^S_{s=1} H^{(n)}_{N,s}$, $SD^S(H^{(n)}_N) := \sqrt{ \frac{1}{S-1}(\sum^S_{s=1}H^{(n)}_{N,s} -  \mathbb{E}^S(H^{(n)}_N))^2}$, and $H^{(n)}_{N,s}$ is computed as in $H^{(n)}_N$ but using the scaled residuals from the $s$-th random sample of $X_{1,s},\ldots,X_{n,s} \stackrel{i.i.d.}{\sim} N(0, \textbf{I}_p)$; that is, $\mathbb{E}^S(H^{(n)}_N)$ and $SD^S(H^{(n)}_N)$ are the estimated mean and standard deviation of $H^{(n)}_N$  from $S$ independent Monte Carlo simulations. Similar definitions and interpretation apply to $\mathbb{E}^S(D^{(n)}_N)$ and $SD^S(D^{(n)}_N)$. Since we shall find the critical values of the test statistic using simulation, these estimates can be obtained as a byproduct at the same time.

For testing univariate normality, we base on the fact that $\Lambda''(t) = 1$ for all $t \in \mathbb{R}$ if $X \sim N(0, 1)$, and define its empirical and discrete approximate:
\begin{equation}\label{eq:univariate_TS}
	U^{(n)}_N := \sum^N_{l=1}((\Lambda^{(n)})''(t_l) - 1)^2,
\end{equation}
where for some $t_1,\ldots,t_N \in \mathbb{R}$,
\begin{equation*}
	(\Lambda^{(n)})''(t) := \frac{M^{(n)}(t)(M^{(n)})''(t) - (M^{(n)})'(t))^2}{(M^{(n)})^2(t)},
\end{equation*}
such that we reject $H_0$ for large values of $U^{(n)}_N$.
Under the null hypothesis, $T^{(n)}_N$ is independent of the unknown mean vector $\mu$ and covariance matrix $\Sigma$ for the multivariate case, and $U^{(n)}_N$ is independent of the unknown mean and variance for the univariate case. 
In Section \ref{sect:consistency}, we shall further show that our test is consistent especially when the moment generating function exists and is twice differentiable. In general, one cannot find a test that is uniformly powerful than all the other tests against all alternatives. In Section \ref{sect:simulation}, through an extensive simulation study, we see that our proposed test often has higher powers over various alternative distributions compared with some prevalent common tests as well as some recently formulated tests. In particular, for the univariate case, our test performs the best in various short-tailed symmetric and asymmetric distributions compared with other common normality test, including the well-known and powerful Shapiro-Wilk test (\cite{shapiro1965analysis}).  For the multivariate case, our proposed test outperforms other common tests in many of the short-tailed distributions and some of the long-tailed distributions. For other distributions, the proposed test tends to have powers in between those of different tests.

The organization of the paper is as follows. In Section \ref{sect:asy_null_dist}, we derive the asymptotic distribution of the test statistic $T^{(n)}_N$ under the null hypothesis $H_0$. Consistency of the newly proposed test statistic will be established in Section \ref{sect:consistency}. We provide a comprehensive Monte Carlo simulation study in Section \ref{sect:simulation}. Discussion and conclusion are given in Section  \ref{sect:conclusion}.
Technical proofs and additional simulation results are appended in the section of supplementary materials.

\section{Asymptotic Null Distribution}\label{sect:asy_null_dist}
In this section, we derive the asymptotic distribution of the test statistic $T^{(n)}_N$ under the null hypothesis that each sample $X_i$ has a nondegenerate multivariate normal distribution. Recall that our proposed test statistic is independent of the mean and covariance matrix of $X$. To derive the asymptotic null distribution, it therefore suffices to consider the case when $\mathbb{E}(X) = 0$ and $Var(X) = \textbf{I}_p$.

Let $M^{(n)}_0(t) := \frac{1}{n}\sum^n_{i=1} e^{t^\top X_i}$, $\triangledown M^{(n)}_0(t) := \frac{1}{n}\sum^n_{i=1} X_i e^{t^\top X_i}$ and $H_{M^{(n)}_0}(t) := \frac{1}{n}\sum^n_{i=1}X_i X_i^\top e^{t^\top X_i}$ be the empirical moment generating function, its gradient and Hessian based on $X_i$'s (not the scaled residuals $Z_{n,i}$'s), respectively. The corresponding population versions are $M_0(t) = e^{\|t\|^2/2}$, $\triangledown M_0(t) = e^{\|t\|^2/2} t$ and $H_{M_0}(t) = e^{\|t\|^2/2}(tt^\top +  \textbf{I}_p)$. Define $f_0:\mathbb{R}^p \times \mathbb{R}^p \rightarrow \mathbb{R}$, $f_1: \mathbb{R}^p \times \mathbb{R}^p \rightarrow \mathbb{R}^p$ and $f_2:\mathbb{R}^p \times \mathbb{R}^p \rightarrow \mathbb{R}^{p \times p}$ one by one as follows:
\begin{align*}
	f_0(x;t) &:= e^{\|t\|^2/2} \left\{-t^\top x - \frac{(t^\top x)^2}{2} + \frac{\|t\|^2}{2}\right\}; \\
	f_1(x;t) &:= e^{\|t\|^2/2} \left\{- (\textbf{I}_p + tt^\top)x -  \frac{1}{2} (2 \textbf{I}_p + tt^\top)(x x^\top - \textbf{I}_p)t \right\}; \\
	f_2(x;t) &:= e^{\|t\|^2/2} \bigg\{ - \frac{1}{2} (xx^\top - \textbf{I}_p)(\textbf{I}_p + tt^\top)- x t^\top 
	- \frac{1}{2} (\textbf{I}_p + tt^\top) (xx^\top - \textbf{I}_p)- t x^\top 
	\\
	&\quad  - \frac{1}{2} \sum^p_{h=1}\sum^p_{k=1}  \left( \frac{\partial tt^\top}{\partial t_h} + tt^\top t_h + t_h \textbf{I}_p \right) t_k (x_{h} x_{k}- \mathbbm{1}(h=k)) - (\textbf{I}_p+tt^\top)t^\top x  \bigg\},
\end{align*}
where $\mathbbm{1}(\cdot)$ is the indicator function. Since $\mathbb{E}(X) = 0, \mathbb{E}(XX^\top) = \textbf{I}_p$ and $\mathbb{E}((t^\top X)^2) = t^\top \mathbb{E}(XX^\top)t = \|t\|^2$, straightforward calculation shows that $\mathbb{E}(f_j(X_i;t))=0$ for $j=0,1,2$ and any $t \in \mathbb{R}^p$, where the same notation $0$ is adopted as the zero element in the corresponding high dimensional space if there is no cause of ambiguity. Also, let $Q^h:=\frac{\partial tt^\top}{\partial t_h}$. $Q^h$ is a matrix of $0$'s except for the $h$th row and $h$th column, where the $(h,j)$ element, $Q^h_{hj} = t_j$ for $j=1,\ldots,h-1,h+1,\ldots,p$; while the $(i,h)$ element, $Q^h_{ih}=t_i$ for $i=1,\ldots,h-1,h+1,\ldots,p$, and $Q^h_{hh} = 2 t_h$.
The following lemma first gives the approximation errors due to the use of scaled residuals in the empirical version of the moment generating function and its gradient and Hessian. 
\begin{lemma}\label{lemma:Mn_Mn0}
	Under $H_0$, for any $t \in \mathbb{R}^p$, we have
	\begin{enumerate}[(a)]
		\item $M^{(n)}(t) - M^{(n)}_0(t) = n^{-1}\sum^n_{i=1}f_0(X_i;t) + o_p(n^{-1/2})$;
		
		\item $\triangledown M^{(n)}(t) - \triangledown M^{(n)}_0(t) = n^{-1} \sum^n_{i=1}f_1(X_i;t) + o_p(n^{-1/2})$;
		
		\item  $H_{M^{(n)}}(t) - H_{M^{(n)}_0}(t) = n^{-1}\sum^n_{i=1}f_2(X_i;t) + o_p(n^{-1/2})$.
		
	\end{enumerate}
\end{lemma}

Define $g_1, g_2:\mathbb{R}^p \times \mathbb{R}^p \rightarrow \mathbb{R}^{p \times p}$ by
\begin{align*}
	g_1(x;t) &:=	e^{\|t\|^2/2} \left\{
	e^{t^\top x}(x x^\top + tt^\top - t x^\top - xt^\top - \textbf{I}_p) 	\right\},	\\
	g_2(x;t) &:= e^{\|t\|^2/2} \left\{ f_2(x;t) + (tt^\top + \textbf{I}_p)f_0(x;t) - t f^\top_1(x;t) - f_1(x;t) t^\top - 2 f_0(x;t) \textbf{I}_p \right\}.
\end{align*}
As $\mathbb{E}(e^{t^\top X} XX^\top) = e^{\|t\|/2}(tt^\top + \textbf{I}_p)$ and $\mathbb{E}(e^{t^\top X} X) = t e^{\|t\|^2/2}$, we have $\mathbb{E}(g_1(X_i;t)) = 0$. Also, $\mathbb{E}(g_2(X_i;t)) = 0$ because $\mathbb{E}(f_j(X_i;t))=0$ for $j=0,1,2$.
 To derive the asymptotic distribution of the test statistic $T^{(n)}_N$, we first derive that of $\sqrt{n}(H_{\Lambda^{(n)}}(t) - \textbf{I}_p )$. Since
\begin{equation*}
H_{\Lambda^{(n)}}(t) - \textbf{I}_p = (M^{(n)})^{-2}(t) \cdot \{M^{(n)}(t) H_{M^{(n)}}(t) - \triangledown M^{(n)}(t) (\triangledown M^{(n)}(t))^\top - (M^{(n)})^2(t) \textbf{I}_p \},
\end{equation*}
we establish the following lemma, which write the second term of the right side of the above equation into an average of mean zero terms of independent variable and asymptotic negligible reminder.
\begin{lemma}\label{lemma:Mn_I}
	Under $H_0$, for any $t \in \mathbb{R}^p$, we have
	\begin{enumerate}[(a)]
		\item $M^{(n)}_0(t) H_{M^{(n)}_0}(t) - \triangledown M^{(n)}_0(t) (\triangledown M^{(n)}_0(t))^\top - (M^{(n)}_0)^2(t) \textbf{I}_p = n^{-1}\sum^n_{i=1} g_1(X_i;t) + o_p(n^{-1/2})$;
		
		\item $	M^{(n)}(t) H_{M^{(n)}}(t) - \triangledown M^{(n)}(t) (\triangledown M^{(n)}(t))^\top - (M^{(n)})^2(t) \textbf{I}_p = n^{-1}\sum^n_{i=1} (g_1(X_i;t) +g_2(X_i;t)) + o_p(n^{-1/2})$.

	\end{enumerate}

\end{lemma}

With Lemmas \ref{lemma:Mn_Mn0} and \ref{lemma:Mn_I}, we have the following theorem.
\begin{thm}\label{thm:H_I}
	Under $H_0$, for any $t \in \mathbb{R}^p$, we have
	\begin{equation*}
		\sqrt{n}(H_{\Lambda^{(n)}}(t) - \textbf{I}_p) = \frac{1}{\sqrt{n}}\sum^n_{i=1} h(X_i;t) + o_p(1),
	\end{equation*}
	where $h(x;t) := e^{-\|t\|} (g_1(x;t) + g_2(x;t))$.
\end{thm}

Consider $t_1,\ldots,t_N \in \mathbb{R}^p$ in (\ref{eq:Hn_Dn_t}). For a $p\times p$ matrix $A = (a_{ij})$, we denote $A^v = (a_{11},\ldots,a_{1p},a_{22},\ldots,a_{2p},\ldots,a_{pp})$ to be the vectorization of the upper triangular part of $A$. 
By Theorem \ref{thm:H_I}, the multivariate central limit theorem and Slutsky's theorem, we have
\begin{equation*}
	(	\sqrt{n}(H^v_{\Lambda^{(n)}}(t_1) - \textbf{I}^v_p),\ldots,	\sqrt{n}(H^v_{\Lambda^{(n)}}(t_N) - \textbf{I}^v_p))^\top \stackrel{d}{\rightarrow} N_{p(p+1)N/2}(0, \Sigma_N),
\end{equation*}
where the $(i,j)$-th element of $\Sigma_{N}$ is $\text{Cov}(h^v_{i'}(X;t_l), h^v_{j'}(X;t_{l'}))$, for $i= (l-1) p(p+1)/2  + i', j= (l'-1) p(p+1)/2 + j', l,l' = 1,\ldots,N, i',j' = 1,\ldots,p(p+1)/2$.  By the continuous mapping theorem, we can immediately obtain the limiting distributions of $H^{(n)}_N$, $D^{(n)}_N$, and $T^{(n)}_N$ respectively of, denoted by, $H_N, D_N$ and $T_N$, where
		\begin{equation*}
	T_N := \max  \left\{ 
	\frac{H_N - \mathbb{E}^S(H_N)}{SD^S(H_{N})},
	\frac{D_N - \mathbb{E}^S(D_N)}{SD^S(D_N)}
	\right\},
\end{equation*}
with $\mathbb{E}^S(H_N) := \frac{1}{S}\sum^S_{s=1}H_{N, s}$, $SD^S(H_N) := \sqrt{\frac{1}{S-1}\sum^S_{s=1}(H_{N,s} - \mathbb{E}^S(H_N))}$, and $H_{N,1},\ldots,H_{N,S}$ is a random sample having the same distribution as $H_N$. Similar definition applies to $\mathbb{E}^S(D_N)$ and $SD^S(D_N)$. We do not try to simplify the asymptotic null distribution $T_N$ further because the critical value of the test statistic at any finite sample size can be approximated by simulation.
%
%

\section{Consistency}\label{sect:consistency}
While the scaled residuals $Z_{n,i}$ of a multivariate normal with a general covariance matrix are equal in distribution as the scaled residuals of a standard multivariate normal distribution, such a nice property does not hold with $X$ is in the alternative hypothesis.  To study the convergence of $H_N^{(n)}$ and $D_N^{(n)}$ for a general $X$ with finite second moments, we can assume $E(X)=0$ without loss of generality, and $Var(X)=\Sigma$. Denote $\tilde{X} := \Sigma^{-1/2} X$, $\tilde{M}(t) = \mathbb{E}(e^{t^\top \tilde{X}})$ and $\tilde{\Lambda}(t) := \log \tilde{M}(t)$.
The following theorem establishes the strong limits of $U^{(n)}_N$ for the univariate case, and that of $H^{(n)}_N$ and $D^{(n)}_N$ for the multivariate case, altogether imply the consistency of our test.
\begin{thm}\label{thm:consistency_limit}
	Suppose that the moment generating function of $X$ exists and is twice differentiable. 
	\begin{enumerate}[(a)]
		\item If $X$ is univariate, with probability $1$,
		\begin{equation*}
			\lim_{n \rightarrow \infty} \frac{U^{(n)}_N}{n} = \sum^N_{l=1}(\tilde{\Lambda}''(t_l) - 1)^2.
		\end{equation*}
	\item If $X$ is multivariate, with probability $1$,
	\begin{align*}
		\lim_{n \rightarrow \infty} \frac{H^{(n)}_N}{n} &= \sum^N_{l=1} \sum^p_{i=1}\sum^p_{j > i} (H_{\tilde{\Lambda}^*,ij}(t_l))^2 ,\\
		\lim_{n \rightarrow \infty} \frac{D^{(n)}_N}{n} &= \sum^N_{l=1} \sum^p_{i=1} ( \tilde{D}_{ii}(t_{li}) - 1)^2,
	\end{align*}
where $H_{\tilde{\Lambda}^*}$ and $\tilde{D}_{ii}$ are defined similarly as in (\ref{eq:H_star_lambda_D}).
		\end{enumerate}
\end{thm}
Suppose  that the moment generating function of $X$ exists and is twice differentiable but $X$ is not from $\mathcal{N}_p$. Then, $\tilde{X}$ is not distributed as $N_p(0, \textbf{I}_p)$ and there must exist a point $t^*$ in the neighbourhood of $0$ such that $H_{\tilde{\Lambda}}(t^*) \neq \textbf{I}_p$. Since $H_{\tilde{\Lambda}}(\cdot)$ is continuous, there exists a neighbourhood $\mathcal{O}$ of $t^*$ such that $H_{\tilde{\Lambda}}(t) \neq \textbf{I}_p$ for all $t \in \mathcal{O}$. This implies that  $H_{ \tilde{\Lambda}^*}(t) \neq 0$ or $\tilde{D}(t) \neq \textbf{I}_p$ for all $t \in \mathcal{O}$. As a result, if $\{t_l\}^N_{l=1}$ is chosen using a space-filling design and $N$ is large enough, some $t_l$ will be in $\mathcal{O}$.  Hence, almost surely, $\lim_{n \rightarrow \infty} H^{(n)}_N = \infty$ or $\lim_{n \rightarrow \infty} D^{(n)}_N = \infty$ so that $\lim_{n \rightarrow \infty} T^{(n)}_N = \infty$. In practice, one can also perform the test with a large enough $N$ and see if the $p$-value is stable when $N$ increases.

\section{Simulation Studies}\label{sect:simulation}
We carried out extensive Monte Carlo study to evaluate the finite-sample sizes and powers of our proposed test statistics and compare with several tests in the literature. 
For the univariate case, we compare with 
\begin{enumerate}[(a)]
	\item the Cram\'er-von Mises (CvM) test (\cite{cramer1928composition}, \cite{mises1931wahrscheinlichkeitsrechnung}, and \cite{smirnov1936sui});
	\item the Anderson-Darling (AD) test (\cite{anderson1954test});
	\item the Shapiro-Wilk (SW) test (\cite{shapiro1965analysis});
	\item the Jarque-Bera (JB) test (\cite{jarque1987test});
	\item the Henze-Visagie (HV) test (\cite{henze2020testing}).
\end{enumerate}
The first four tests are well-known and will not be reviewed; see \cite{yap2011comparisons} for a review of these tests. 
For the implementation of CvM and AD, the functions \verb|cvm.test| and \verb|ad.test| in the R package \verb|nortest| are used, respectively. The SW test can be carried out using \verb|Shapiro.test| in the \verb|stats| R package. The JB test is carried out using \verb|jarque.bera.test| from the R package \verb|tseries|.
The HV test is developed based on a system of first-order partial differential equations that characterize the moment generating function of the $p$-variate standard normal distribution. The test statistic is
\begin{equation}\label{eq:HV}
	HV_{n,\gamma} := n \int_{\mathbb{R}^p} \|\triangledown M^{(n)}(t) - tM^{(n)}(t)\|^2 \exp (-\gamma \|t\|^2) dt,
\end{equation}
which can be computed in a closed-form formula; see equation (9) in (\cite{henze2020testing}). $H_0$ is rejected for large values of $HV_{n,\gamma}$. \cite{henze2020testing} recommended $\gamma = 5$ to be used when performing the test based on their numerical results  and we follow this suggestion in our numerical study. The R package \verb|mnt| contains the function \verb|HV| to compute this test statistic. For our proposed test statistic and the HV test statistic, $100,000$ independent replications were used to determine the critical value of the tests. Each size or power estimate is based on $10,000$ replications. We largely follow the choices of alternative distributions considered in \cite{yap2011comparisons}, where various shapes of distributions were considered. The alternative distributions considered here can be classified into symmetric short-tailed distributions, symmetric long-tailed distributions, asymmetric distributions, and mixture of normal distributions. 

We first define additional notation for some of these distributions. Denote GLD$(\lambda_1,\lambda_2,\lambda_3,\lambda_4)$ to be the genearlized lambda distribution proposed by \cite{ramberg1974approximate}, which is a four-parameter generalization of the two-parameter Tukey's Lambda family of distribution (\cite{hastings1947low}). The percentile function of GLD$(\lambda_1,\lambda_2,\lambda_3,\lambda_4)$ is given as
\begin{equation*}
	Q(y) = \lambda_1 + \frac{y^{\lambda_3} - (1-y)^{\lambda_4}}{\lambda_2}, \quad \text{where } 0 \leq y \leq 1,
\end{equation*}
where $\lambda_1$ is the location parameter, $\lambda_2$ is the scale parameter and $\lambda_3$ and $\lambda_4$ are the shape parameters. The density function of GLD$(\lambda_1,\lambda_2,\lambda_3,\lambda_4)$ at $x = Q(y)$ is 
\begin{equation*}
	f(x) = \frac{\lambda_2}{\lambda_3 y^{\lambda_3 - 1} - \lambda_4(1-y)^{\lambda_4-1}}.
\end{equation*}
The truncated normal distribution of a normal distribution with mean $\mu$ and standard deviation $\sigma$ truncated to the interval $(a,b)$ is denoted by Trunc$(a, b,\mu,\sigma)$. The scale-contaminated normal distribution, denote by ScConN$(p, b)$, is a mixture of two normal distributions with probability $p$ from a normal distribution $N(0, b^2)$ and probability $1-p$ from $N(0, 1)$. LoConN$(p, a)$ denotes the distribution of a mixture of two normal distributions with probability $p$ from a normal distribution with mean $a$ and variance $1$ and with probability $1-p$ from a standard normal distribution.

The symmetric short-tailed distributions include $U(0, 1)$, Beta$(0.5,0.5)$, Beta$(2,2)$, $\allowbreak \text{GLD}(0,1,0.25,0.25)$, GLD$(0,1,0.5,0.5)$, GLD$(0,1,0.75,0.75)$, GLD$(0,1,1.25,1.25)$, $\allowbreak \text{Trunc}(-2,2,0,1)$, $\allowbreak \text{Trunc}(-3,3,0,2)$, and $\text{Trunc}(-2,2,0,2)$. The symmetric long-tailed distributions include Laplace, logistic, GLD$(0, 1,-0.1,-0.1)$, GLD$(0,1,-0.15,-0.15)$, $t(5)$, $t(10)$, and $t(15)$. The asymmetric distributions include exp$(1)$, lognormal$(0, 0.5)$, Gamma$(4, 5)$, Beta$(2, 1)$, Beta$(3, 2)$, Weibull$(3, 1)$, Pareto$(1, 3)$, $\chi^2(4)$, $\chi^2(10)$, and $\chi^2(20)$. The normal mixtures considered are ScConN$(0.2, 5)$, ScConN$(0.05, 5)$, LoConN$(0.5, 3)$, and LoConN$(0.5, 2)$.

In the simulation study, the set of points $\{t_1,\ldots, t_N\}$ is chosen randomly from $\{t \in \mathbb{R}^p: \|t\| \leq R\}$. The upper panel of Figure \ref{figure:uni_B_M_dependence} shows the empirical reject proportions of our test in the univariate case for the $32$ distributions in the alternative hypothesis with different values $R$ when $N$ is fixed at $500$ when the sample size is $50$. It can be seen that the results are similar for $R$ increasing from $3$ to $10$. The lower panel of Figure \ref{figure:uni_B_M_dependence} fixed $R = 3$ with different values of $N$. We can see that different values of $N$ result in similar reject proportions. Similar results were obtained with different sets of random points.

Figures \ref{figure:uni_result1} and \ref{figure:uni_result2} compare our proposed tests with other normality tests at different sample sizes. 
 It can be seen that our test outperformed other tests, including the Shapiro-Wilk test, which is often considered as the best univariate normality test, when the true distribution has a bounded support and tend to have performance in between those of other tests in other distributions. Under $H_0$ (results not shown in the figures), the sizes of our tests are close to $0.05$, as our test statistic is distribution-free under $H_0$ and the critical value is determined using simulation.

For the multivariate case, we compare our proposed test statistic with the following tests: the energy test of \cite{szekely2005new}, the Henze-Visagie (HV) test (defined in (\ref{eq:HV})), the Henze–Jim\'enez-Gamero (HJ) test (\cite{henze2019new}, the Henze-Zirkler (HZ) test (\cite{henze1990class}) and the Mardia's test (\cite{mardia1970measures}) based on skewness (MS) and kurtosis (MK). A brief description of these tests is as follows. \cite{szekely2005new} proposed the test statistic
\begin{equation*}
	\mathcal{E}_n := n \left( \frac{2}{n}\sum^n_{i=1}\mathbb{E}\|Z_{n,i}-X\| - \frac{2 \Gamma((p+1)/2)}{\Gamma(p/2)} - \frac{1}{n^2}\sum^n_{i,j=1}\|Z_{n,i} - Z_{n,j}\| \right),
\end{equation*}
where the first expectation is taking with respect to $X$, which follows $N(0, \textbf{I}_p)$, and
\begin{equation*}
	\mathbb{E}(\|a-X\|) = \frac{\sqrt{2} \Gamma(\frac{p+1}{2})}{\Gamma(\frac{p}{2})} +
	\sqrt{ \frac{2}{\pi}} \sum^\infty_{k=0}\frac{(-1)^k}{k!2^k} \frac{\|a\|^{2k+2}}{(2k+1)(2k+2)} \frac{ \Gamma( \frac{p+1}{2}) \Gamma(k + \frac{3}{2})}{\Gamma(k+\frac{p}{2}+1)}.
\end{equation*}
The test using $\mathcal{E}_n$ is known as the energy test and we make use of the function \verb|mvnorm.etest| in the R package \verb|engergy| to compute the test statistic and its $p$-value. \cite{szekely2005new} concluded that the energy test is a powerful omnibus test having relatively good power against general alternatives compared with other tests. 
The Henze–Jim\'enez-Gamero test is based on the test statistic
\begin{equation*}
	HJ_{n, \beta} := n \int_{\mathbb{R}^p} (M^{(n)}(t) - M(t))^2 \exp (- \beta\| t\|^2) dt,
\end{equation*} 
where the test rejects $H_0$ for large values of $HJ_{n,\beta}$. We include HV and HJ tests for comparison because our test is also based on characterization of normality using a system of partial differential equations involving the moment generating function. It will be of interest to compare the performance of these tests. The Henze-Zirkler test is based on empirical characteristic function of the scaled residuals:
\begin{equation*}
	HZ_{n,\gamma} := (2\pi\gamma^2)^{-p/2} \int_{\mathbb{R}^d} \left| \frac{1}{n}\sum^n_{j=1}\exp(i t^\top Z_{n,j} ) - \exp \left( - \frac{\|t\|^2}{2} \right) \right|^2 \exp \left( - \frac{\|t\|^2}{2\gamma^2} \right)dt.
\end{equation*}
We compute the test statistics for HV, HJ and HZ tests using the functions \verb|HV|, \verb|HJ|, and \verb|HZ| in the package \verb|mnt|, respectively.
Mardia's test for multinormality based on sample skewness rejects $H_0$ for large values of $b_{1,p}$, where
\begin{equation*}
	b_{1,p} := \frac{1}{n^2}\sum^n_{i,j=1}(Z^\top_{n,i} Z_{n,j})^3.
\end{equation*}
The sample kurtosis is given by
\begin{equation*}
	b_{2, p} := \frac{1}{n}\sum^n_{i=1}\|Z_{n,i}\|^4.
\end{equation*}
See \cite{mardia1970measures} for their limiting null distributions. To perform the tests based on the sample skewness and kurtosis, we use the function \verb|mult.norm| from the R package \verb|QuantPsyc|. All the critical values for these multivariate tests are determined from $100,000$ independent replications under the null hypothesis.

For the distributions in the alternative hypothesis, we consider multivariate distributions with independent components  where the marginal distributions are the $32$ alternative distributions considered in the univariate case, except for the normal mixture. In addition, we consider dependent multivariate distributions with normal marginals, where the dependence structure is generated from some common copulas, including the Clayton, Gumbel, Frank, Ali–Mikhail–Haq (AMH), and $t$ copulas (see \cite{nelsen2007introduction}). Except for the $t$ copula, the other copulas are parameterized by one parameter. For $t$ copula, we use the notation tCopula($\rho$, df), where $\rho$ and df denote the parameter in the exchangeable dispersion structure and degrees of freedom, respectively. We also consider the setting where the dependence structure is generated from the Gaussian copula but the marginal distributions are non-normal.

Figure \ref{figure:multi_B_M_dependence} shows the performance of the proposed test with different values of $R$ and $N$ when $n = 50$. The results are similar and the test with $R = 3$ performed the best over the distributions considered. For the values of $N$, as long as it is not too small, increasing $N$ will not increase the power of the test. Figures \ref{figure:multi_result1} and \ref{figure:multi_result2} compares the empirical reject proportions of the proposed test to other tests with different dimensions and when $n = 50$, $R = 3$ and $N = 500$. We see that our test performs the best in many different multivariate distributions. Under $H_0$, all the tests considered have around $0.05$ sizes as the critical values can be determined by simulation.

More details of the simulation results for both the univariate and multivariate cases are given in the Appendix.

\begin{figure}[H]
	\centering
	\includegraphics[width = 14cm]{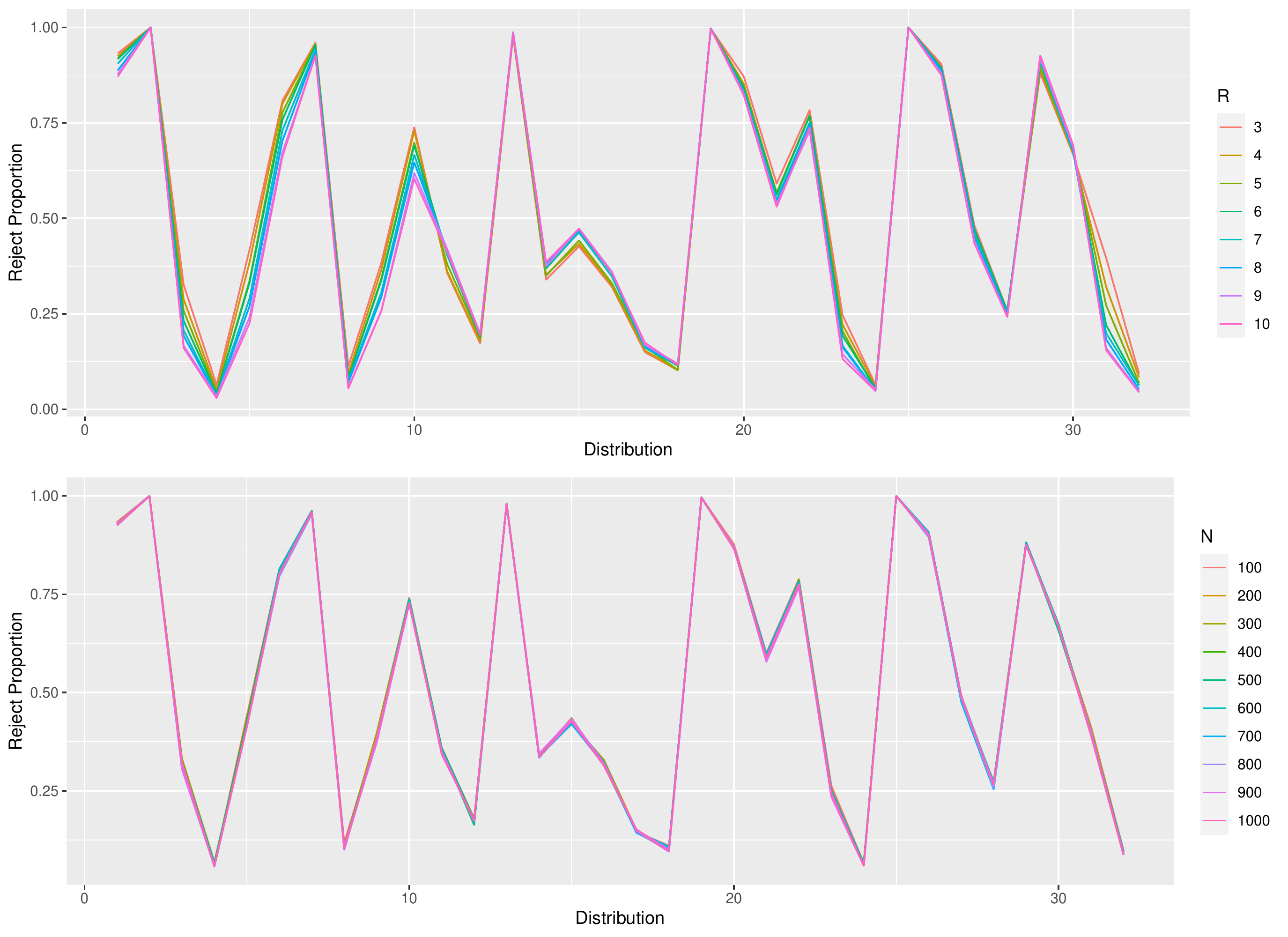}
	\caption{Univariate test under alternatives when $n = 50$. The upper panel shows the empirical reject proportions under the $32$ alternative distributions when $N = 500$ with different values of $R$. The lower panel shows  the empirical reject proportions under the $32$ alternative distributions when $R = 3$ with different values of $N$.}
	\label{figure:uni_B_M_dependence}
\end{figure}

\begin{figure}[H]
	\centering
	\includegraphics[width = 16cm]{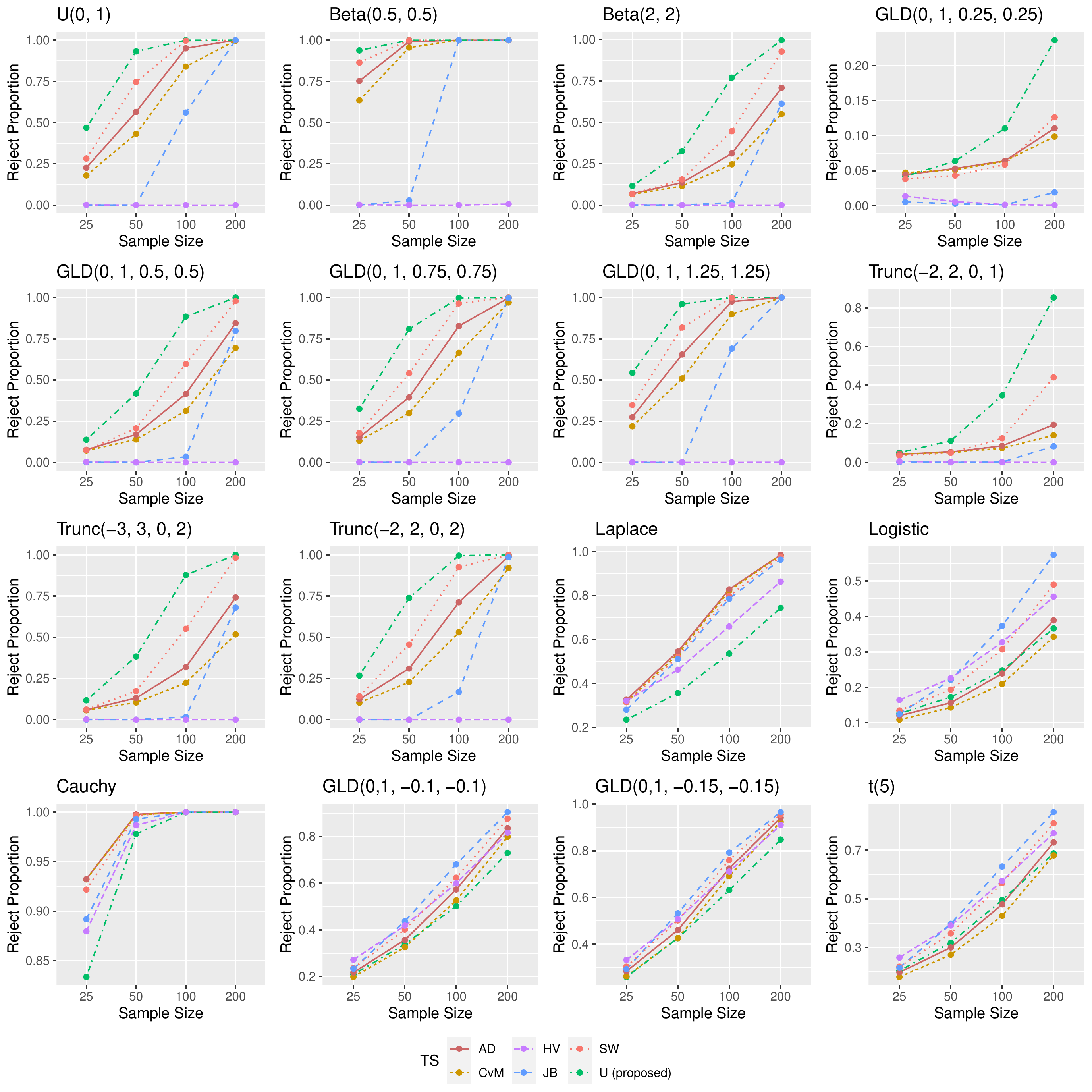}
	\caption{Univariate test under different alternative distributions. Here $R = 3$ and $N = 500$ for our test statistic.}
	\label{figure:uni_result1}
\end{figure}

\begin{figure}[H]
	\centering
	\includegraphics[width = 16cm]{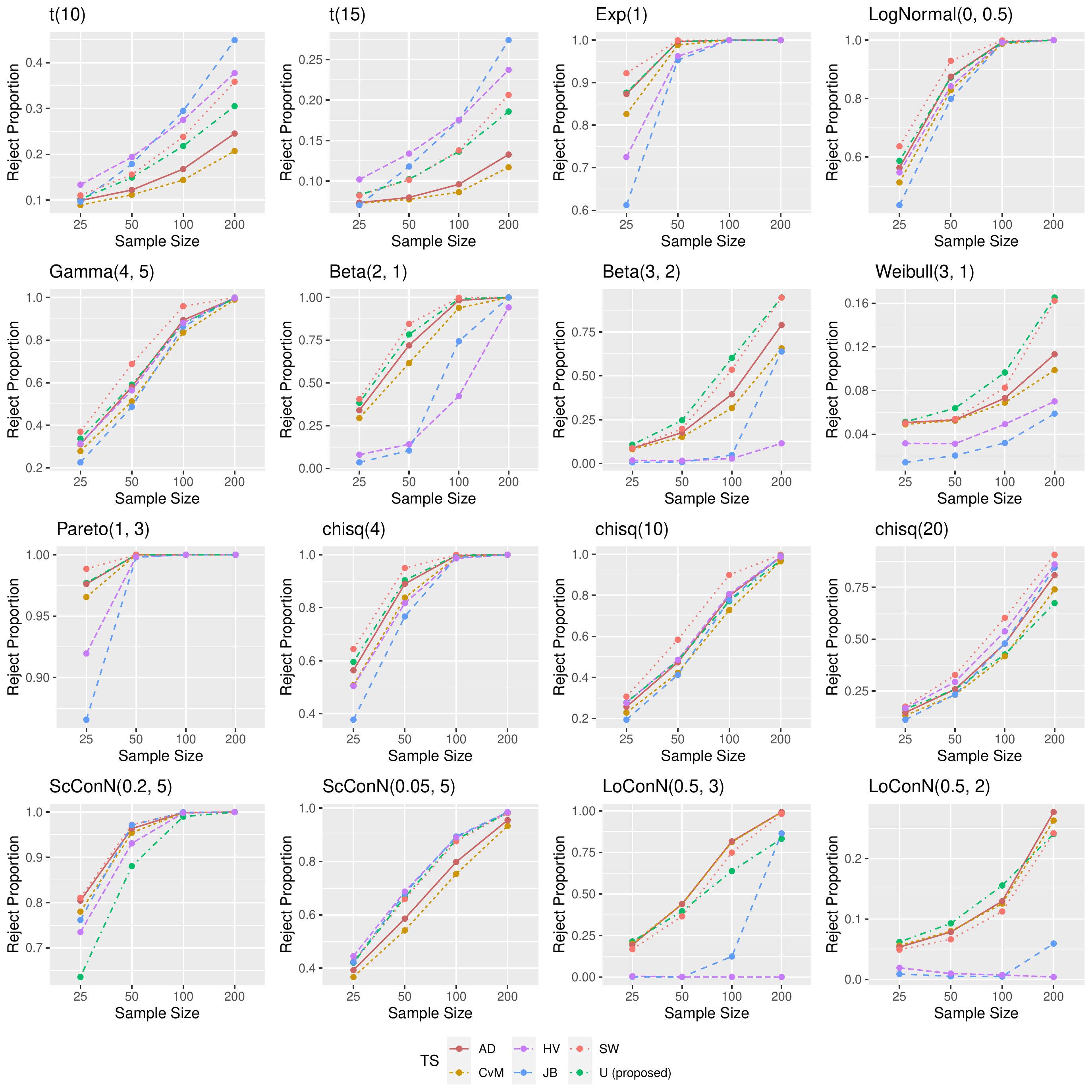}
	\caption{Univariate test under different alternative distributions (continued). Here $R = 3$ and $N = 500$ for our test statistic.}
	\label{figure:uni_result2}
\end{figure}

\begin{figure}[H]
	\centering
	\includegraphics[width = 14cm]{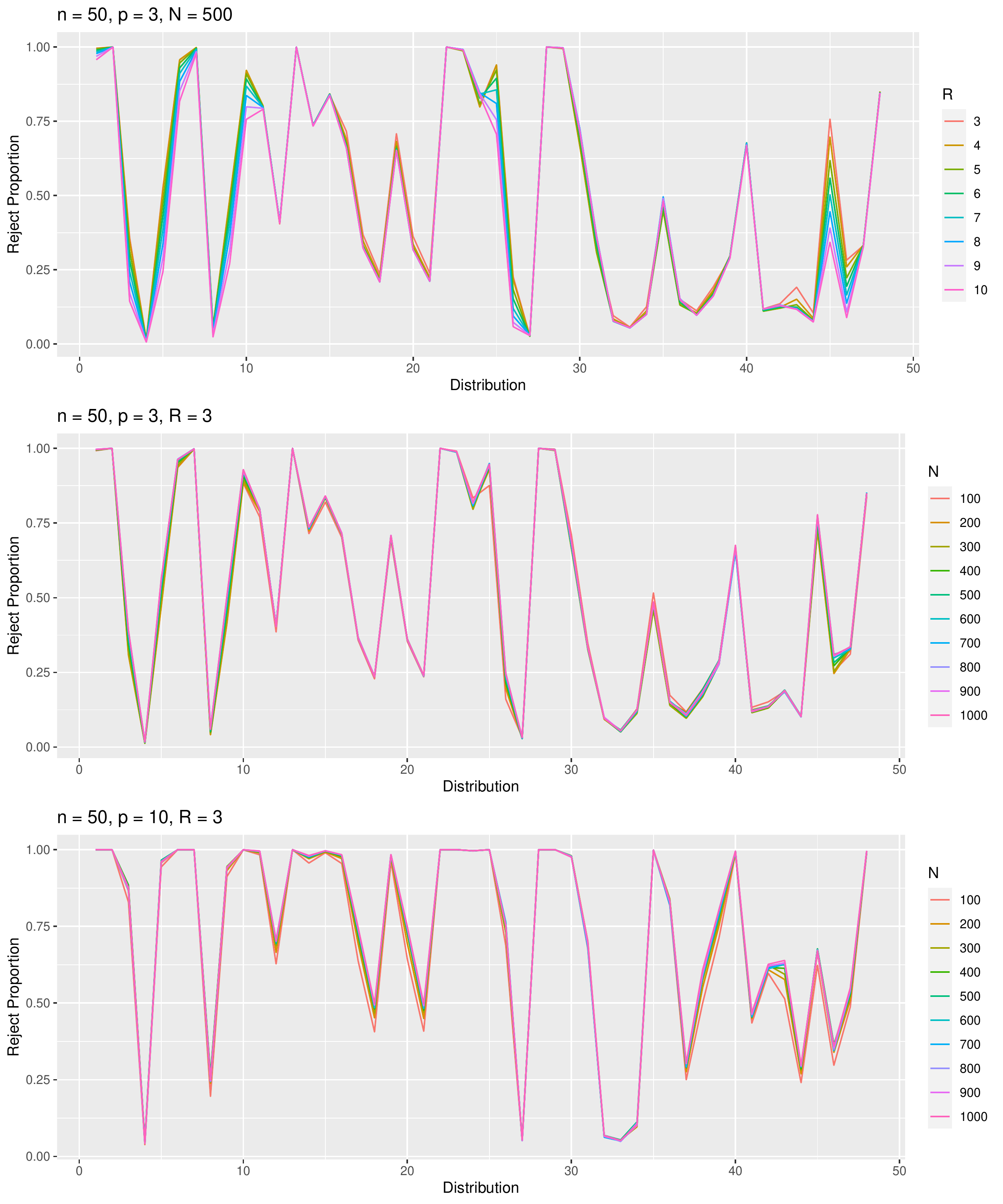}
	\caption{Multivariate test under alternatives when $n = 50$. The upper panel shows the empirical reject proportions under the $48$ alternative distributions when $N = 500$ and $p = 3$ with different values of $R$. 
		The middle panel shows  the empirical reject proportions when $R = 3$ and $p = 3$ with different values of $N$.
		The lower panel shows  the empirical reject proportions when $R = 3$ and $p = 10$ with different values of $N$.}
\label{figure:multi_B_M_dependence}
\end{figure}

\begin{figure}[H]
	\centering
	\includegraphics[width = 16cm]{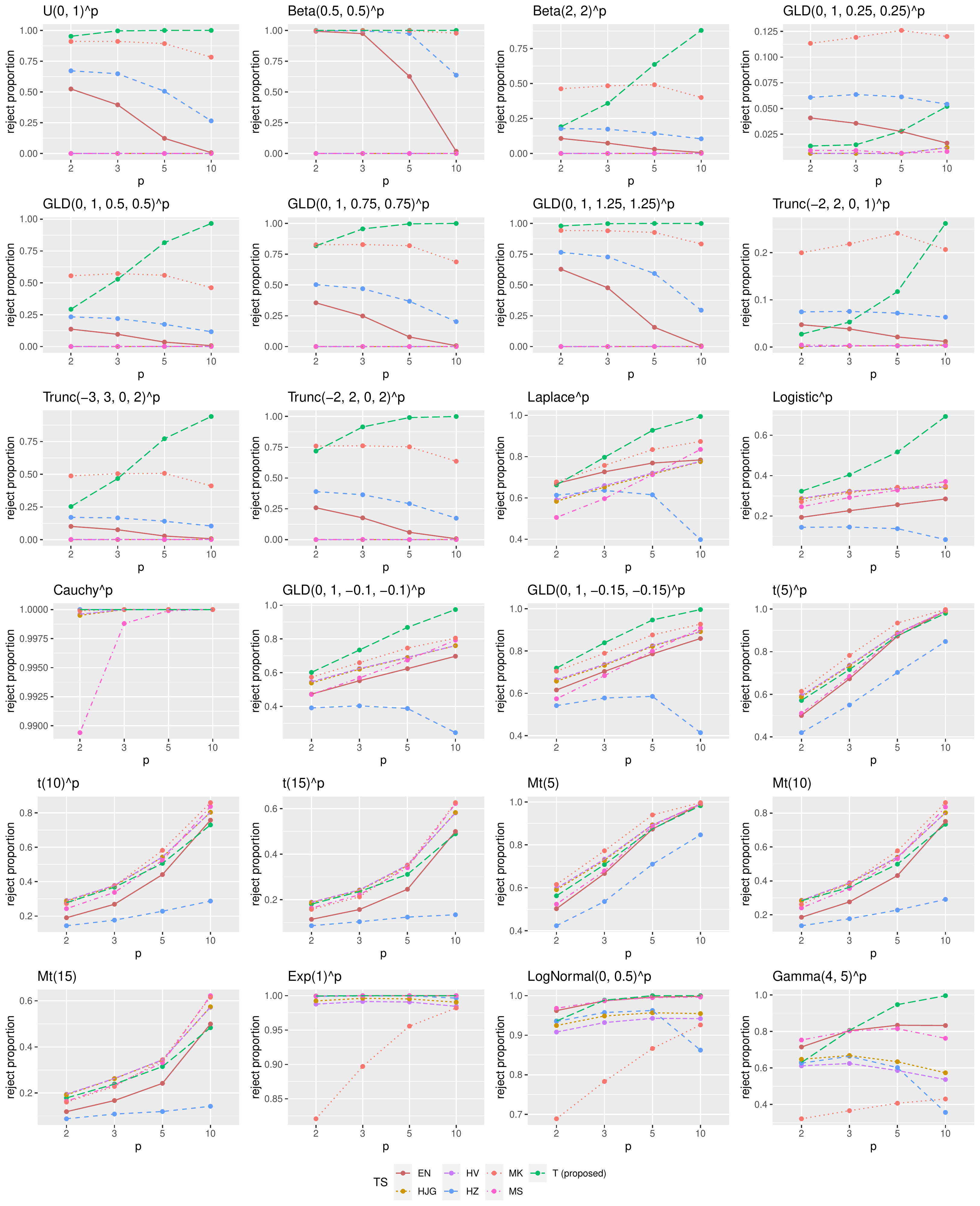}
	\caption{Multivariate test under different alternative distributions and different dimensions $p$ when $n = 50$. Here $R = 3$ and $N = 500$ for our test statistic.}
\label{figure:multi_result1}
\end{figure}

\begin{figure}[H]
	\centering
	\includegraphics[width = 16cm]{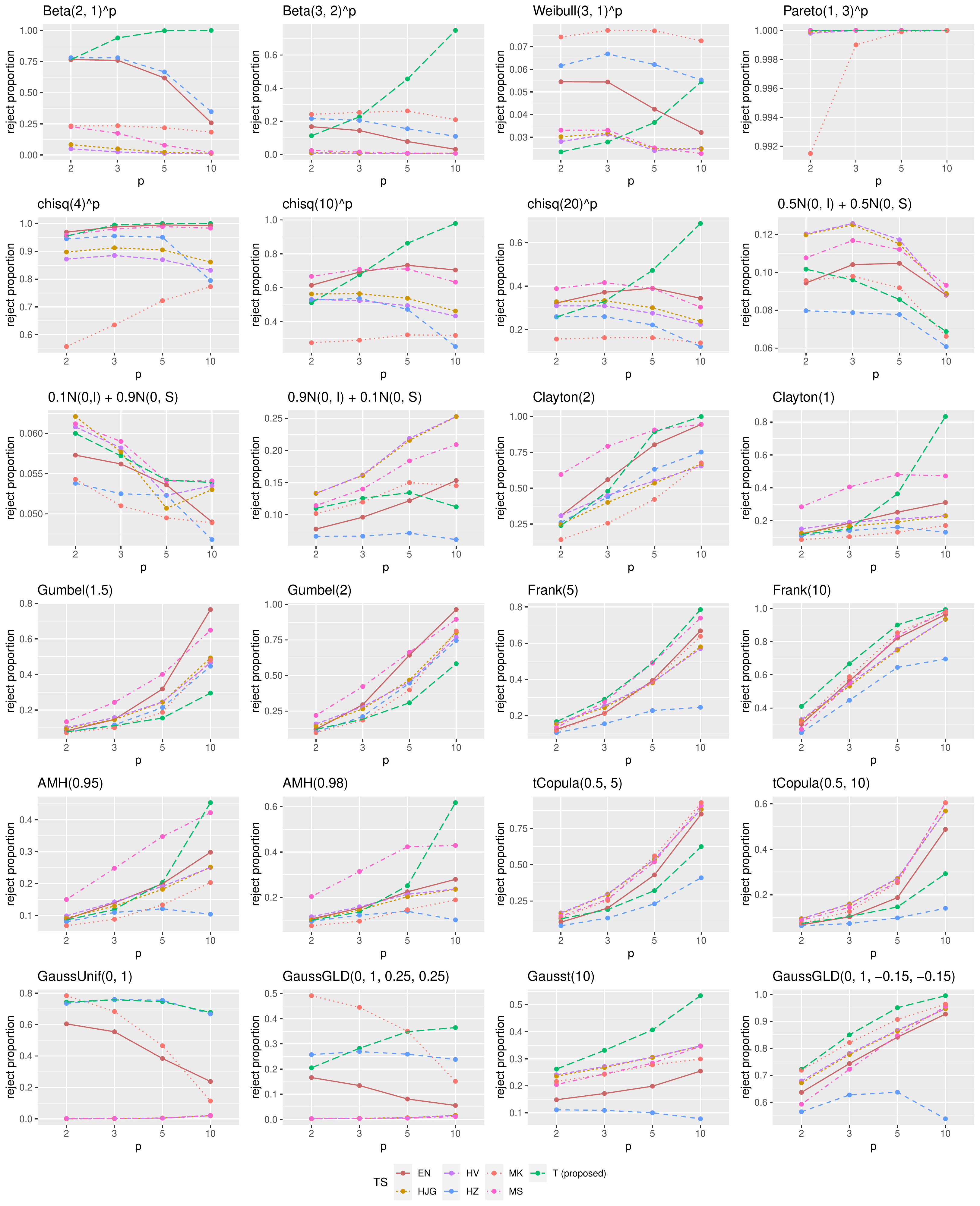}
	\caption{Multivariate test under different alternative distributions and different dimensions $p$ when $n = 50$. Here $R = 3$ and $N = 500$ for our test statistic. $S$ is a matrix with all $0.5$'s in the off-diagonal entries and $1$'s in the diagonal entries.}
	\label{figure:multi_result2}
\end{figure}

\section{Conclusion}\label{sect:conclusion}
In this article, a novel class of tests for multivariate normality is proposed. Our extensive Monte Carlo study suggested that our test is more powerful in many alternatives compared with existing common tests. 
We suggest performing our test with $R = 3$ and $N = 500$, where the set of points can be simulated uniformly from the ball centred at $0$ with radius $R$.

Another possible class of tests can be obtained replacing the moment generating function with the characteristic function in the definition of the cumulant generating function. The finite-sample performance of such a test is undergoing investigation. The idea of combining the dependence structure and marginal information can also be applied to other goodness-of-fit problems.

\section*{Acknowledgement}
Chan and Tang were partially funded by the US National Institutes of Health grant R01HL122212. Hok Kan Ling acknowledges the support by NSERC
Grant RGPIN/03124-2021.
The fourth author acknowledges financial support from Hong Kong General Research Fund Grants HKGRF-14300319 “Shape-Constrained Inference: Testing for Monotonicity” and HKGRF-14301321“General Theory for Infinite Dimensional Stochastic Control: Mean Field and Some Classical Problems.

\section{Appendix} 
To facilitate the proofs, we first define additional notations and describe some of their properties. Let
\begin{equation}\label{eq:Delta_ni}
	\Delta_{n,i} := Z_{n,i} - X_i = (S^{-1/2}_n - \textbf{I}_p)X_i - S^{-1/2}_n\overline{X}_n.	
\end{equation}
By Taylor's theorem, for some $|\theta_{n,i}(t)| \leq 1$, we have
\begin{equation}\label{eq:Taylor_etDelta}
	e^{t^\top \Delta_{n,i}} = 1 + t^\top \Delta_{n,i} + \frac{1}{2} (t^\top \Delta_{n,i})^2 e^{\theta_{n,i}(t) t^\top \Delta_{ni}}.
\end{equation}
Under $H_0$, by (2.13) of \cite{henze1997new}, we have
\begin{equation}\label{eq:S_squareroot_I}
	\sqrt{n}(S^{-1/2}_n - \textbf{I}_p) = - \frac{1}{2\sqrt{n}} \sum^n_{i=1} (X_i X_i^\top -  \textbf{I}_p) + O_p(n^{-1/2}),
\end{equation}
which is $O_p(1)$ by the central limit theorem because $\mathbb{E}(X_i X_i^\top - \textbf{I}_p) = 0$. Define $\|\cdot\|_2$ to be the spectral norm of a matrix. From (\ref{eq:Delta_ni}), we have
\begin{equation}\label{eq:Delta_bound}
	\|\Delta_{n,i}\| \leq  \|(S^{-1/2}_n - \textbf{I}_p)\|_2 \cdot \|X_i\| + \|S^{-1/2}_n\|_2 \cdot  \|\overline{X}_n\|.
\end{equation}
This inequality (\ref{eq:Delta_bound}) together with (\ref{eq:S_squareroot_I}) and the fact that  $\max_{i=1,\ldots,n}\|X_i\| = O_p(\log n)$ (see Proposition A.1 in \cite{henze2019characterizations}) imply that
\begin{equation}\label{eq:t_Delta_ni2}
	\max_{i=1,\ldots,n} \| \Delta_{n,i}\| = O_p(n^{-1/2}\log n).
\end{equation}
Recall that $M_0(t)$ denotes the moment generating function of $X \sim N_p(0, \textbf{I}_p)$. By the strong law of large numbers, for each $t \in \mathbb{R}^p$, $M_{n,0}(t) \stackrel{a.s.}{\rightarrow}  M(t) = e^{\|t\|^2/2}$, $\triangledown M_{n,0}(t) \stackrel{a.s.}{\rightarrow} \triangledown M_0(t) = t e^{\|t\|^2/2}$, and $H_{M_{n,0}}(t) \stackrel{a.s.}{\rightarrow} H_{M_0}(t) = (tt^\top + \textbf{I}_p) e^{\|t\|^2/2}$. Finally, $\frac{1}{n}\sum^n_{i=1} e^{t^\top X_i} X_i X_i^\top X_{ih} \stackrel{a.s.}{\rightarrow} \mathbb{E}(e^{t^\top X} X X^\top X_h)$, where $X_h$ is the $h$-th component of $X$ and
\begin{align}
	\mathbb{E}(e^{t^\top X} X X^\top X_h) &	 = \frac{\partial}{\partial t_h} \mathbb{E}(e^{t^\top X} XX^\top) 
	= \frac{\partial}{\partial t_h}\left\{ (tt^\top + \textbf{I}_p) e^{\|t\|^2/2}\right\} \nonumber  \\
	&= \frac{\partial tt^\top}{\partial t_h} e^{\|t\|^2/2} + tt^\top t_h e^{\|t\|^2/2} + t_h e^{\|t\|^2/2}\textbf{I}_p \nonumber \\
	&=e^{\|t\|^2/2} \left( \frac{\partial tt^\top}{\partial t_h} + tt^\top t_h + t_h \textbf{I}_p \right). \label{eq:EXXXh}
\end{align}
\subsection{Proofs for Section \ref{sect:asy_null_dist}}

\begin{proof}[Proof of Lemma \ref{lemma:Mn_Mn0}]
	\begin{enumerate}[(a)]
		\item 
		By  (\ref{eq:Delta_ni}) and (\ref{eq:Taylor_etDelta}),
		\begin{align}
			\sqrt{n}(M^{(n)}(t) - M^{(n)}_0(t)) 
			&= \frac{1}{\sqrt{n}}\sum^n_{i=1} e^{t^\top X_i} (e^{t^\top \Delta_{n,i}} - 1)  \nonumber \\
			& = \frac{1}{\sqrt{n}}\sum^n_{i=1} e^{t^\top X_i} \left\{ t^\top \Delta_{n,i} + \frac{1}{2}(t^\top \Delta_{n,i})^2 e^{\theta_{n,i}(t) t^\top \Delta_{n,i}} \right\} \nonumber \\
			&=: A_{1n} + A_{2n} + A_{3n} + A_{4n}, \label{eq:Lemma_1_a_decompose}
		\end{align}
		where
		\begin{align*}
			A_{1n} &:= \frac{1}{\sqrt{n}}\sum^n_{i=1} e^{t^\top X_i} t^\top (S^{-1/2}_n - \textbf{I}_p)X_i;\\
			A_{2n} &:= -\frac{1}{\sqrt{n}}\sum^n_{i=1} e^{t^\top X_i} t^\top (S^{-1/2}_n -  \textbf{I}_p) \overline{X}_n; \\
			A_{3n} &:= -\frac{1}{\sqrt{n}}\sum^n_{i=1} e^{t^\top X_i} t^\top \overline{X}_n;\\
			A_{4n} &:= \frac{1}{\sqrt{n}} \sum^n_{i=1} e^{t^\top X_i}\frac{1}{2}(t^\top \Delta_{n,i})^2 e^{\theta_{n,i}(t) t^\top \Delta_{n,i}}.
		\end{align*}
		For $A_{1n}$, we have	
		\begin{align}
			A_{1n}	&= \frac{1}{n}\sum^n_{i=1} e^{t^\top X_i} X_i^\top  \sqrt{n}(S^{-1/2}_n -\textbf{I}_p)t \nonumber \\
			&= \left( \frac{1}{n}\sum^n_{i=1} e^{t^\top X_i} X_i^\top \right) \left\{ - \frac{1}{2\sqrt{n}} \sum^n_{j=1} (X_j X_j^\top - \textbf{I}_p) + O_p(n^{-1/2}) \right\} t \nonumber  \\
			&= \left( \frac{1}{n}\sum^n_{i=1} e^{t^\top X_i} X_i^\top \right) \left\{ - \frac{1}{2\sqrt{n}} \sum^n_{j=1} (X_j X_j^\top - \textbf{I}_p)  \right\} t  + o_p(1) \nonumber  \\
			& = e^{\|t\|^2/2} t^\top \left\{ - \frac{1}{2\sqrt{n}} \sum^n_{j=1} (X_j X_j^\top - \textbf{I}_p)  \right\}  t  + o_p(1) \nonumber  \\
			&= - \frac{e^{\|t\|^2/2}}{2\sqrt{n}} \sum^n_{i=1}\left\{ (t^\top X_i)^2 - \|t\|^2 \right\}  + o_p(1), \label{eq:Lemma_1_a_decompose_A1n}
		\end{align}
		where the second equality follows from (\ref{eq:S_squareroot_I}) and the second last equality follows from the weak law of large numbers.
		For $A_{2n}$, by (\ref{eq:S_squareroot_I}), 
		\begin{align}
			A_{2n}				&= \bigg(- \frac{1}{n} \sum^n_{i=1} e^{t^\top X_i} t^\top \bigg)\sqrt{n}(S^{-1/2}_n - \textbf{I}_p) \overline{X}_n \nonumber  \\
			&= O_p(1) O_p(1) O_p(n^{-1/2}) = o_p(1). \label{eq:Lemma_1_a_decompose_A2n}
		\end{align}		
		For $A_{3n}$, by the weak law of large numbers,
		\begin{align}
			A_{3n} 	&= \bigg( - e^{\|t\|^2/2}t^\top + o_p(1) \bigg) \sqrt{n} \cdot \overline{X}_n = -e^{\|t\|^2/2} \frac{1}{\sqrt{n}} \sum^n_{i=1} t^\top X_i + o_p(1). \label{eq:Lemma_1_a_decompose_A3n}
		\end{align}
		By (\ref{eq:t_Delta_ni2}), we have
		\begin{align}		
			A_{4n}	&\leq \frac{\sqrt{n}}{2} \|t\|^2 \max_{i=1,\ldots,n} \|\Delta_{n,i} \|^2 e^{\|t\| \max_{i=1,\ldots,n}\|\Delta_{n,i}\|} \left( \frac{1}{n} \sum^n_{i=1} e^{t^\top X_i} \right) \nonumber \\
			&= \sqrt{n} O_p(n^{-1}(\log n)^2 ) O_p(1) O_p(1)= o_p(1). \label{eq:Lemma_1_a_decompose_A4n}
		\end{align}
	The result then follows from (\ref{eq:Lemma_1_a_decompose})-(\ref{eq:Lemma_1_a_decompose_A4n}).
		
		\item 		By  (\ref{eq:Delta_ni}) and (\ref{eq:Taylor_etDelta}),
		\begin{align}
			&	\sqrt{n}(\triangledown M^{(n)}(t) - \triangledown M^{(n)}_{0}(t)) \nonumber \\
			&= \frac{1}{\sqrt{n}}\sum^n_{i=1}(Z_{n,i} - X_i) e^{t^\top Z_{n,i}} + \frac{1}{\sqrt{n}}\sum^n_{i=1} X_i \left( e^{t^\top Z_{n,i}} - e^{t^\top X_i} \right)\nonumber \\
			&= \frac{1}{\sqrt{n}}\sum^n_{i=1} \Delta_{n,i} e^{t^\top X_i} e^{t^\top \Delta_{n,i}} + \frac{1}{\sqrt{n}}\sum^n_{i=1} X_i e^{t^\top X_i} (e^{t^\top \Delta_{n,i}} - 1) \nonumber \\
			&= B_{1n} + B_{2n} + B_{3n} + B_{4n},\label{eq:B1n_B2n_B3n_B4n}
		\end{align}
		where
		\begin{align*}
			B_{1n} &:= \frac{1}{\sqrt{n}}\sum^n_{i=1} (S^{-1/2}_n - \textbf{I}_p) X_i e^{t^\top X_i} e^{t^\top \Delta_{n,i}};\\
			B_{2n} &:= - \frac{1}{\sqrt{n}}\sum^n_{i=1} (S^{-1/2}_n -  \textbf{I}_p) \overline{X}_n e^{t^\top X_i} e^{t^\top \Delta_{n,i}};\\
			B_{3n} &:= - \frac{1}{\sqrt{n}}\sum^n_{i=1} \overline{X}_n e^{t^\top X_i} e^{t^\top \Delta_{n,i}};\\
			B_{4n} &:= \frac{1}{\sqrt{n}}\sum^n_{i=1} X_i e^{t^\top X_i} \left\{ t^\top \Delta_{n,i} + \frac{1}{2}(t^\top \Delta_{n,i})^2 e^{\theta_{n,i}(t) t^\top\Delta_{n,i}} \right\}.
		\end{align*}
		By (\ref{eq:S_squareroot_I}) and the weak law of large numbers,
		\begin{align}
			B_{1n} &= \frac{1}{n} \sum^n_{i=1} \left\{ - \frac{1}{2\sqrt{n}} \sum^n_{j=1}(X_j X_j^\top - \textbf{I}_p) \right\} X_i e^{t^\top X_i} + o_p(1)  \nonumber \\
			&= - \frac{1}{2\sqrt{n}} \sum^n_{i=1} (X_i X_i^\top - \textbf{I}_p) \left( t e^{\|t\|^2/2} + o_p(1) \right) + o_p(1)\nonumber  \\
			&=- \frac{e^{\|t\|^2/2}}{2\sqrt{n}} \sum^n_{i=1} (X_i X_i^\top t - t)    + o_p(1). \label{eq:B1n}
		\end{align}
		By (\ref{eq:S_squareroot_I}),
		\begin{align}
			\|B_{2n}\| & \leq \left \|\frac{1}{n}\sum^n_{i=1}e^{t^\top X_i}\right\| \|S^{-1/2}_n - \textbf{I}_p\|_2 \|\sqrt{n} \cdot \overline{X}_n\| e^{ \|t\| \max_{i=1,\ldots,n} \|\Delta_{n,i}\| } \nonumber  \\
			&= O_p(1) O_p(n^{-1/2}) O_p(1) O_p(1) = o_p(1). \label{eq:B2n}
		\end{align}	
		By the weak law of large numbers,
		\begin{align}\label{eq:B3n}
			B_{3n} &= - \frac{1}{n}\sum^n_{i=1} e^{t^\top X_i} \sqrt{n} \cdot \overline{X}_n + o_p(1) 
			= - \frac{e^{\|t\|^2/2} }{\sqrt{n}}\sum^n_{i=1}X_i + o_p(1).
		\end{align}
		Finally, by  (\ref{eq:Delta_ni}) 
		\begin{align}
			B_{4n} &= \frac{1}{\sqrt{n}}\sum^n_{i=1} X_i e^{t^\top X_i} t^\top \Delta_{n,i}  + o_p(1) \nonumber \\
			&= \frac{1}{\sqrt{n}}\sum^n_{i=1} X_i e^{t^\top X_i} t^\top (S^{-1/2}_n - \textbf{I}_p) X_i - \frac{1}{\sqrt{n}}\sum^n_{i=1} X_i e^{t^\top X_i} t^\top (S^{-1/2}_n - \textbf{I}_p) \overline{X}_n \nonumber \\
			& \quad - \frac{1}{\sqrt{n}}\sum^n_{i=1} X_i e^{t^\top X_i} t^\top \overline{X}_n +o_p(1)\nonumber \\
			&=\frac{1}{\sqrt{n}}\sum^n_{i=1} e^{t^\top X_i} X_i  X^\top_i (S^{-1/2}_n - \textbf{I}_p) t - t e^{\|t\|^2/2} \frac{1}{\sqrt{n}} \sum^n_{i=1} t^\top X_i + o_p(1)\nonumber \\
			&= - \frac{e^{\|t\|^2/2}}{2\sqrt{n}} (\textbf{I}_p + tt^\top) \sum^n_{i=1} (X_i X_i^\top - \textbf{I}_p) t - \frac{e^{\|t\|^2/2}}{\sqrt{n}} \sum^n_{i=1} t t^\top X_i +o_p(1), \label{eq:B4n}
		\end{align}
	where the last equality follows from  (\ref{eq:S_squareroot_I}) and the weak law of large numbers.
		Combining (\ref{eq:B1n_B2n_B3n_B4n}), (\ref{eq:B1n}), (\ref{eq:B2n}), (\ref{eq:B3n}), and (\ref{eq:B4n}), we have
			\begin{align*}
			&	\sqrt{n}(\triangledown M^{(n)}(t) - \triangledown M^{(0)}_{0}(t)) \\
			&= \frac{e^{\|t\|^2/2}}{\sqrt{n}} \sum^n_{i=1} \left\{ - \frac{1}{2} (X_iX_i^\top t - t) - X_i -\frac{1}{2}(\textbf{I}_p + tt^\top) (X_i X_i^\top - \textbf{I}_p)t - tt^\top X_i \right\} + o_p(1)\\
			&= \frac{e^{\|t\|^2/2}}{\sqrt{n}} \sum^n_{i=1} \left\{- (\textbf{I}_p +tt^\top)X_i - \frac{1}{2} (2 \textbf{I}_p + tt^\top) (X_iX_i^\top - \textbf{I}_p)t \right\} + o_p(1).
		\end{align*}
			
		\item 	Note that
		\begin{align*}
			Z_{n,i}Z_{n,i}^\top - X_i X_i^\top &= 	(\Delta_{n,i}+X_i)(\Delta_{n,i}+ X_i)^\top - X_i X_i^\top\\
			&=\Delta_{n,i} \Delta^\top_{n,i} + \Delta_{n,i} X^\top_i + X_i \Delta^\top_{n,i}.
		\end{align*}
		Thus,
		\begin{align*}
			&	\sqrt{n}( H_{M^{(n)}}(t) - H_{M^{(n)}_{0}}(t))\\
			&= \frac{1}{\sqrt{n}}\sum^n_{i=1} (Z_{n,i} Z_{n,i}^\top - X_i X_i^\top) e^{t^\top X_i} e^{t^\top \Delta_{n,i}} + \frac{1}{\sqrt{n}}\sum^n_{i=1} X_i X_i^\top e^{t^\top X_i} (e^{t^\top \Delta_{n,i}} - 1)\\
			&=C_{1n} + C_{2n} + C_{3n} +C_{4n},
		\end{align*}
		where 
		\begin{align*}
			C_{1n} &:= \frac{1}{\sqrt{n}}\sum^n_{i=1} \Delta_{n,i}\Delta_{n,i}^\top e^{t^\top X_i} e^{t^\top \Delta_{n,i}};\\
			C_{2n} &:= \frac{1}{\sqrt{n}}\sum^n_{i=1} \Delta_{n,i}X^\top_i e^{t^\top X_i} e^{t^\top \Delta_{n,i}};\\
			C_{3n} &:= \frac{1}{\sqrt{n}}\sum^n_{i=1} X_i\Delta_{n,i}^\top e^{t^\top X_i} e^{t^\top \Delta_{n,i}};\\
			C_{4n} &:= \frac{1}{\sqrt{n}} \sum^n_{i=1} X_i X_i^\top e^{t^\top X_i} \left\{ t^\top \Delta_{n,i} + \frac{1}{2}(t^\top \Delta_{n,i})^2e^{\theta_{n,i}(t) t^\top \Delta_{n,i}} \right\}.
		\end{align*}
		For $C_{1n}$, by (\ref{eq:t_Delta_ni2}),
		\begin{align*}
			\|C_{1n}\|_2 &\leq \sqrt{n} \left| \frac{1}{n}\sum^n_{i=1}e^{t^\top X_i} \right| \max_{i=1,\ldots,n} \|\Delta_{n,i}\|^2 e^{\| t\| \max_{i=1,\ldots,n} \|\Delta_{n,i}\|} \\
			&= \sqrt{n} O_p(1) O_p(n^{-1} (\log n)^2)O_p(1) = o_p(1).
		\end{align*}
		For $C_{2n}$, by (\ref{eq:S_squareroot_I}),
		\begin{align*}
			C_{2n} &=\frac{1}{\sqrt{n}}\sum^n_{i=1} (S^{-1/2}_n - \textbf{I}_p) X_i X_i^\top e^{t^\top X_i} - \frac{1}{\sqrt{n}}\sum^n_{i=1} (S^{-1/2}_n - \textbf{I}_p) \overline{X}_n X_i^\top e^{t^\top X_i} \\
			& \quad - \frac{1}{\sqrt{n}} \sum^n_{i=1} \overline{X}_n X_i^\top e^{t^\top X_i} + o_p(1)\\
			&= \frac{1}{n}\sum^n_{i=1} \left\{- \frac{1}{2\sqrt{n}} \sum^n_{j=1}(X_j X_j^\top - \textbf{I}_p) + O_p(n^{-1/2}) \right\} X_i X_i^\top e^{t^\top X_i} + O_p(n^{-1/2})\\
			& \quad - \frac{1}{\sqrt{n}} \sum^n_{i=1} X_i \left( t^\top e^{\|t\|^2/2} + o_p(1) \right) + o_p(1)\\
			&=- \frac{1}{2\sqrt{n}}\sum^n_{i=1} (X_i X_i^\top - \textbf{I}_p) e^{\|t\|^2/2} (\textbf{I}_p+ tt^\top + o_p(1)) - \frac{e^{\|t\|^2/2}}{\sqrt{n}}\sum^n_{i=1} X_i t^\top + o_p(1) \\
			&= -\frac{e^{\|t\|^2/2}}{2\sqrt{n}}\sum^n_{i=1}(X_i X_i^\top - \textbf{I}_p) (\textbf{I}_p + tt^\top) -
			\frac{e^{\|t\|^2/2}}{\sqrt{n}}\sum^n_{i=1} X_i t^\top + o_p(1).
		\end{align*}
		Since $C_{3n} = C_{2n}^\top$, we have
		\begin{equation*}
			C_{3n} = -\frac{e^{\|t\|^2/2}}{2\sqrt{n}}\sum^n_{i=1}  (\textbf{I}_p + tt^\top)(X_i X_i^\top - \textbf{I}_p) -
			\frac{e^{\|t\|^2/2}}{\sqrt{n}}\sum^n_{i=1} t X^\top_i  + o_p(1).
		\end{equation*}
		For $C_{4n}$, 		
		\begin{align*}
			C_{4n} &= \frac{1}{\sqrt{n}} \sum^n_{i=1} X_i X_i^\top e^{t^\top X_i} t^\top (S^{-1/2}_n - \textbf{I}_p) X_i - \frac{1}{\sqrt{n}}\sum^n_{i=1} X_i X_i^\top e^{t^\top X_i} t^\top (S^{-1/2}_n - \textbf{I}_p)\overline{X}_n \\
			&\quad - \frac{1}{\sqrt{n}} \sum^n_{i=1} X_i X_i^\top e^{t^\top X_i} t^\top \overline{X}_n + o_p(1)\\
			&=\frac{1}{n}\sum^n_{i=1} e^{t^\top X_i} X_i X_i^\top X^\top_i \left\{ - \frac{1}{2\sqrt{n}} \sum^n_{j=1}(X_j X_j^\top - \textbf{I}_p) +O_p(n^{-1/2}) \right\} t + O_p(n^{-1/2})\\
			& \quad - \frac{e^{\|t\|^2/2}}{\sqrt{n}} (\textbf{I}_p + tt^\top) \sum^n_{i=1} t^\top X_i + o_p(1)\\
			&=\frac{1}{n}\sum^n_{i=1} e^{t^\top X_i} X_i X_i^\top X^\top_i \left\{ - \frac{1}{2\sqrt{n}} \sum^n_{j=1}(X_j X_j^\top - \textbf{I}_p) \right\} t  - \frac{e^{\|t\|^2/2}}{\sqrt{n}} (\textbf{I}_p + tt^\top) \sum^n_{i=1} t^\top X_i \\
			& \quad + o_p(1).
		\end{align*}
		Denote $\tilde{C}_n:= - \frac{1}{2\sqrt{n}} \sum^n_{j=1}(X_j X_j^\top - \textbf{I}_p)$.
		Note that
		\begin{align*}
			&	\frac{1}{n}\sum^n_{i=1} e^{t^\top X_i} X_i X_i^\top X^\top_i \tilde{C}_n t\\
			&= \frac{1}{n}\sum^n_{i=1} e^{t^\top X_i} X_i X_i^\top \sum^p_{h=1}\sum^p_{k=1} X_{ih} \tilde{C}_{n,hk}t_k\\
			&= \sum^p_{h=1}\left\{ \left( \frac{1}{n}\sum^n_{i=1} e^{t^\top X_i} X_i X_i^\top X_{ih} \right) \sum^p_{k=1} t_k \tilde{C}_{n,hk}\right\}\\
			&= \sum^p_{h=1}\left\{ \mathbb{E}(e^{t^\top X} XX^\top X_h) \sum^p_{k=1} t_k \tilde{C}_{n,hk}\right\} +o_p(1)\\
			&= - \frac{1}{2\sqrt{n}}\sum^n_{i=1} \left\{ \sum^p_{h=1} \sum^p_{k=1} \mathbb{E}(e^{t^\top X} XX^\top X_h) t_k (X_{ih} X_{ik} - \mathbbm{1}(h=k)) \right\} + o_p(1).
		\end{align*}
		Thus, by (\ref{eq:EXXXh}), 
		\begin{align*}
			C_{4n} 
			&= - \frac{e^{\|t\|^2/2}}{2\sqrt{n}}\sum^n_{i=1} \left\{ \sum^p_{h=1} \sum^p_{k=1} \left( \frac{\partial tt^\top}{\partial t_h} + tt^\top t_h + t_h \textbf{I}_p \right) t_k (X_{ih} X_{ik} - \mathbbm{1}(h=k)) \right\} \\
			& \quad - \frac{e^{\|t\|^2/2}}{\sqrt{n}} (\textbf{I}_p + tt^\top) \sum^n_{i=1} t^\top X_i + o_p(1).
		\end{align*}
	\end{enumerate}
	
\end{proof}

\begin{proof}[Proof of Lemma \ref{lemma:Mn_I}]
	\begin{enumerate}[(a)]
		\item We can write
		\begin{equation*}
			M^{(n)}_{0}(t) H_{M^{(n)}_{0}}(t) - \triangledown M^{(n)}_{0}(t) (\triangledown M^{(n)}_{0}(t))^\top - (M^{(n)}_{0}(t))^2 \textbf{I}_p = E_{1n} + E_{2n} + E_{3n},
		\end{equation*}
		where
		\begin{align*}
			E_{1n} &:= M_0(t)H_{M^{(n)}_{0}}(t) - \triangledown M_0(t) (\triangledown M^{(n)}_{0}(t))^\top - M_0(t) M^{(n)}_{0}(t) \textbf{I}_p; \\
			E_{2n} &:= \{M^{(n)}_{0}(t) - M_0(t)\} H_{M_0}(t) - \{\triangledown M^{(n)}_{0}(t) - \triangledown M_0(t)\} (\triangledown M_{0}(t))^\top \\
			& \quad - \{M^{(n)}_{0}(t) - M_0(t)\} M_{0}(t) \textbf{I}_p; \\
			E_{3n} &:= \{M^{(n)}_{0}(t) - M_0(t)\} \{H_{M^{(n)}_{0}}(t) - H_{M_{0}}(t)\} \\
			& \quad - \{\triangledown M^{(n)}_{0}(t) -  \triangledown M_0(t)\} \{ \triangledown  M^{(n)}_{0}(t)  - \triangledown  M_0(t) \}^\top\\
			&\quad  - \{M^{(n)}_{0}(t) - M_0(t)\} \{M^{(n)}_{0}(t) - M_0(t)\} \textbf{I}_p.
		\end{align*}
		For $E_{1n}$, we have
		\begin{align*}
			E_{1n} &= \frac{1}{n}\sum^n_{i=1} \bigg( e^{\|t\|^2/2}X_iX_i^\top e^{t^\top X_i} - e^{\|t\|^2/2}t X_i^\top e^{t^\top X_i} - e^{\|t\|^2/2}e^{t^\top X_i} \textbf{I}_p \bigg)\\
			&= \frac{e^{\|t\|^2/2}}{n} \sum^n_{i=1} e^{t^\top X_i} (X_i X_i^\top - t X_i^\top - \textbf{I}_p).
		\end{align*}
		For $E_{2n}$, we have
		\begin{align*}
			E_{2n}&:= \frac{1}{n}\sum^n_{i=1}(e^{t^\top X_i} - e^{\|t\|^2/2})\{ H_{M_0}(t) - M_0(t)\textbf{I}_p\} \\
			& \quad -  \frac{1}{n}\sum^n_{i=1}(X_i e^{t^\top 
				X_i} - t e^{\|t\|^2/2}) (\triangledown M_0(t))^\top	\\
			&= \frac{e^{\|t\|^2/2}}{n} \sum^n_{i=1} \left\{ (e^{t^\top X_i} - e^{\|t\|^2/2}) tt^\top  - (X_it^\top e^{t^\top X_i} - tt^\top e^{\|t\|^2/2}) \right\} \\
			&= \frac{e^{\|t\|^2/2}}{n} \sum^n_{i=1} \left\{ e^{t^\top X_i} (tt^\top - X_it^\top)  \right\}.
		\end{align*}
		By Lemma \ref{lemma:Mn_Mn0},  $E_{3n} = O_p(n^{-1})$. Thus, the result follows.
		
		\item 
		Denote $F_{1n} := M^{(n)}(t) H_{M^{(n)}}(t)$, $F_{2n} := \triangledown M^{(n)}(t) (\triangledown M^{(n)}(t))^\top$ and $F_{3n} :=  (M^{(n)}(t))^2 \textbf{I}_p$. Then,
		\begin{equation*}
			M^{(n)}(t) H_{M^{(n)}}(t) - \triangledown M^{(n)}(t) (\triangledown M^{(n)}(t))^\top - (M^{(n)}(t))^2 \textbf{I}_p = F_{1n} - F_{2n} - F_{3n}.
		\end{equation*}
		Recall that $\mathbb{E}(f_j(X_i;t)) = 0$, by the central limit theorem $n^{-1}\sum^n_{i=1}f_j(X_i;t) = O_p(n^{-1/2})$ for $j=0,1,2$. By Lemma \ref{lemma:Mn_Mn0}, 
		\begin{align*}
			F_{1n} &= \left( M^{(n)}_{0}(t) + \frac{1}{n}\sum^n_{i=1}f_0(X_i;t) + o_p(n^{-1/2})\right)
			\left(H_{M^{(n)}_{0}}(t) + \frac{1}{n}\sum^n_{i=1}f_2(X_i;t) + o_p(n^{-1/2})\right) \\
			&= M^{(n)}_{0}(t) H_{M^{(n)}_{0}}(t)+ M^{(n)}_{0}(t) \left(\frac{1}{n}\sum^n_{i=1} f_2(X_i;t)\right) + \left(\frac{1}{n}\sum^n_{i=1} f_0(X_i;t)\right) H_{M^{(n)}_{0}}(t) + o_p(n^{-1/2}) \\
			&= M^{(n)}_{0}(t) H_{M^{(n)}_{0}}(t)+ M_{0}(t) \left(\frac{1}{n}\sum^n_{i=1} f_2(X_i;t)\right) + \left(\frac{1}{n}\sum^n_{i=1} f_0(X_i;t)\right) H_{M_{0}}(t) + o_p(n^{-1/2}),
		\end{align*}
		where the last equality holds as $(M^{(n)}_{0}(t) - M_0(t))\left( n^{-1}\sum^n_{i=1}f_2(X_i;t) \right)= O_p(n^{-1/2}) O_p(n^{-1/2}) = O_p(n^{-1})$ and $\left( n^{-1}\sum^n_{i=1}f_0(X_i;t) \right)(H_{M^{(n)}_{0}} (t) - H_{M_0}(t)) =  O_p(n^{-1/2}) O_p(n^{-1/2}) = O_p(n^{-1})$, by the central limit theorem. Thus,
		\begin{equation}\label{eq:lemma_Bn1}
			F_{1n} = M^{(n)}_{0}(t) H_{M^{(n)}_{0}}(t) + \frac{e^{\|t\|^2/2}}{n}\sum^n_{i=1} \left\{
			f_2(X_i;t) + (tt^\top + \textbf{I}_p) f_0(X_i;t)	\right\} + o_p(n^{-1/2}).
		\end{equation}
		Similarly, by  Lemma \ref{lemma:Mn_Mn0}, 
		\begin{align*}
F_{2n} 	&= \left(\triangledown M^{(n)}_{0}(t) + \frac{1}{n}\sum^n_{i=1}f_1(X_i;t) + o_p(n^{-1/2})\right)\\
				&\quad \quad 
				\left( \triangledown M^{(n)}_{0}(t) + \frac{1}{n}\sum^n_{i=1}f_1(X_i;t) + o_p(n^{-1/2})\right)^\top \\
			&= \triangledown M^{(n)}_{0}(t) (\triangledown M^{(n)}_{0}(t))^\top + \triangledown M^{(n)}_{0}(t) \left(\frac{1}{n}\sum^n_{i=1}f_1(X_i;t)^\top \right) \\
			& \quad + \left( \frac{1}{n}\sum^n_{i=1}f_1(X_i;t)\right) (\triangledown M^{(n)}_{0}(t))^\top + o_p(n^{-1/2}) \\
			&=\triangledown M^{(n)}_{0}(t) (\triangledown M^{(n)}_{0}(t))^\top + \triangledown M_{0}(t) \left( \frac{1}{n}\sum^n_{i=1}f_1(X_i;t)^\top \right)\\
			& \quad + \left( \frac{1}{n}\sum^n_{i=1}f_1(X_i;t)\right) (\triangledown M_{0}(t))^\top   + o_p(n^{-1/2}).
		\end{align*}
		Thus,
		\begin{equation}\label{eq:lemma_Bn2}
			F_{n2} = \triangledown M^{(n)}_{0}(t) (\triangledown M^{(n)}_{0}(t))^\top + \frac{e^{\|t\|^2/2}}{n}\sum^n_{i=1} \left(
			t f^\top_1(X_i;t) + f_1(X_i;t)t^\top 
			\right) + o_p(n^{-1/2}).
		\end{equation}
		Using the same argument, we have
		\begin{align}
			F_{3n} &=
			\left( M^{(n)}_{0}(t) + \frac{1}{n}\sum^n_{i=1}f_0(X_i;t) + o_p(n^{-1/2}) \right)^2\textbf{I}_p \nonumber \\
			&=\left( \left(M^{(n)}_{0}(t)\right)^2 + \frac{2 M^{(n)}_{0}(t)}{n} \sum^n_{i=1}f_0(X_i;t) + o_p(n^{-1/2})\right) \textbf{I}_p  \nonumber\\
			&=	\left( M^{(n)}_{0}(t)\right)^2  \textbf{I}_p + \frac{e^{\|t\|^2/2}}{n}\sum^n_{i=1}(2 f_0(X_i;t) \textbf{I}_p ) + o_p(n^{-1/2}) \label{eq:lemma_Bn3}.
		\end{align}
		The result follows by combining (\ref{eq:lemma_Bn1}), (\ref{eq:lemma_Bn2}), and (\ref{eq:lemma_Bn3}) and part (a) of this lemma.
		
	\end{enumerate}
	
\end{proof}

\begin{proof}[Proof of Theorem \ref{thm:H_I}]
	By Lemma \ref{lemma:Mn_Mn0} (a) and Lemma \ref{lemma:Mn_I} (b), we have
	\begin{align*}
		&	\sqrt{n}(H_{\Lambda_n}(t) - \textbf{I}_p)\\
		&= (M^{(n)}(t))^{-2} \sqrt{n} \{ M^{(n)}(t)H_{M^{(n)}}(t) - \triangledown M^{(n)}(t) (\triangledown  M^{(n)}(t))^\top - (M^{(n)})^2(t) \textbf{I}_p\}\\
		&= M^{-2}_0(t) \sqrt{n} \{ M^{(n)}(t)H_{M^{(n)}}(t) - \triangledown M^{(n)}(t) (\triangledown  M^{(n)}(t))^\top - (M^{(n)})^2(t) \textbf{I}_p\} \\
		&\quad + \{(M^{(n)}(t))^{-2}-M^{-2}_0(t)\} \sqrt{n} \{ M^{(n)}(t)H_{M^{(n)}}(t) - \triangledown  M^{(n)}(t) (\triangledown  M^{(n)}(t))^\top - (M^{(n)})^2(t) \textbf{I}_p\}\\
		&= \frac{1}{\sqrt{n}}\sum^n_{i=1} h(X_i;t) +o_p(1) +O_p(n^{-1/2})O_p(1).
	\end{align*}
	Thus, the results follow.
\end{proof}

	\subsection{Proofs for Section \ref{sect:consistency}}
The following lemma is a generalization of Proposition 5 in  \cite{henze2020testing}, where that proposition corresponds to the case when $A = \textbf{I}_p$.
	\begin{lemma}\label{lemma:matrix_sqrt_conv}
		Let $A_n$ be a sequence of $p \times p$ symmetric positive definite matrices, $A$ be a $p \times p$ symmetric positive definite matrix, and $b_n$ be an increasing sequence of positive real numbers satisfying $\lim_{n\to\infty}b_n = \infty$. Suppose that $			\lim_{n\to\infty}b_n\|A_n - A\|_2 = 0$, then
		\begin{equation*}
			\lim_{n\to\infty}b_n\|A_n^{-1/2} - A^{-1/2}\|_2 = 0. 		
		\end{equation*}
		
	\end{lemma}

	\begin{proof}[Proof of Lemma \ref{lemma:matrix_sqrt_conv}]
 Note that 
		\begin{align*}
			A_n - A =&~ (A_n^{1/2} + A^{1/2})(A_n^{1/2} - A^{1/2})\\
			=&~ (A_n^{1/2} + A^{1/2})A_n^{1/2}(A^{-1/2} - A_n^{-1/2})A^{1/2}.
		\end{align*}
		Thus, we have
		\begin{equation}\label{eq:matrix_key_inq}
			b_n\|A^{-1/2} - A_n^{-1/2}\|_2 
			\leq \|A_n^{-1/2}\|_2 \cdot
			\|(A_n^{1/2} + A^{1/2})^{-1}\|_2 \cdot
			\|A^{-1/2}\|_2 \cdot 	 (b_n \|A_n - A\|_2).
		\end{equation}
		Since $\lim_{n \rightarrow \infty}b_n\|A_n - A\|_2 = 0$, it remains to show that the other terms on the RHS of  (\ref{eq:matrix_key_inq}) are bounded. Clearly, $\|A^{-1/2}\|_2 < \infty$.
		Denote $\lambda_{\min}(B)$ to be the smallest eigenvalue of a matrix $B$. Let $\lambda_1 := \lambda_{\min}(A) > 0$ since $A$ is positive definite. Since $\|\cdot\|_2$ is the spectral norm, by Weyl's inequality, we have
		\begin{align}\label{eq:40} 
			\|(A_n^{1/2} + A^{1/2})^{-1}\|_2
			= \frac{1}{\lambda_{\min}\left(A^{1/2}_n + A^{1/2}\right)} 
			\leq \frac{1}{\lambda_{\min}(A^{1/2}_n) + \lambda_1^{1/2}} \leq \frac{1}{\lambda^{1/2}_1} < \infty. 
		\end{align}
		Let $F_n := A - A_n$. 	Since $\lim_{n \rightarrow \infty} b_n\|A_n - A\|_2 = 0$, there exists $n_0 \in \mathbb{N}$ such that for all $n \geq n_0$, $\|F_n\|_2 = \|A_n - A\|_2 \leq \lambda_1/2$. 
		For all $n \geq n_0$, by Weyl's inequality again,
		\begin{align*}
			\|A_n^{-1}\|_2
			=&~ \frac{1}{\lambda_{\min}(A_n)} =  \frac{1}{\lambda_{\min}(A - F_n)} \leq \frac{1}{\lambda_1 + \lambda_{\min}(-F_n)} \\
			\leq&~ \frac{1}{\lambda_1 - \|F_n\|_2 } \leq  \frac{1}{\lambda_1 - \lambda_1/2} =  \frac{2}{\lambda_1} <\infty.
		\end{align*}
		This implies that $\|A_n^{-1/2}\|_2\leq \sqrt{2/\lambda_1} < \infty$ for all $n \geq n_0$.
	\end{proof}
	
Let $\tilde{X}_i := \Sigma^{-1/2}X_i$ for $i=1,\ldots,n$. Denote $\tilde{M}^{(n)}(t) := \frac{1}{n}\sum^n_{i=1} e^{t^\top \tilde{X}_i}$, $\triangledown \tilde{M}^{(n)}(t) := \frac{1}{n}\sum^n_{i=1} \tilde{X}_i e^{t^\top \tilde{X}_i}$ and $H_{\tilde{M}^{(n)}}(t) := \frac{1}{n}\sum^n_{i=1} \tilde{X}_i \tilde{X} _i e^{t^\top \tilde{X}_i}$.
	\begin{lemma}\label{lemma:alternative_M}
		Suppose that the moment generating function of $X$ exists and is twice differentiable. 	We have 
		\begin{enumerate}[(a)]
			\item $M^{(n)}(t) - \tilde{M}^{(n)}(t) \stackrel{a.s.}{\rightarrow} 0$;
			
			\item $\| \triangledown M^{(n)}(t) - \triangledown  \tilde{M}^{(n)}(t) \| \stackrel{a.s.}{\rightarrow} 0$;
			
			\item $\|H_{M^{(n)}}(t) - H_{\tilde{M}^{(n)}}(t)\|_2 \stackrel{a.s.}{\rightarrow} 0$.
			
		\end{enumerate}
	\end{lemma}
	
	\begin{proof}[Proof of Lemma \ref{lemma:alternative_M}]
		Define
		\begin{equation*}
			\tilde{\Delta}_{n,i} := Z_{n,i} - \tilde{X}_i = (S^{-1/2}_n - \Sigma^{-1/2})X_i - S^{-1/2}_n \overline{X}_n.
		\end{equation*}
		Let $\xi_n := \max_{i=1,\ldots,n}\| \tilde{\Delta}_{n,i}\|$. Then,
		\begin{equation}\label{eq:xi_n_ineq}
			\xi_n \leq n^{1/4} \|S^{-1/2}_n - \Sigma^{-1/2}\|_2 \cdot n^{-1/4} \max_{i=1,\ldots,n}\|X_j\| + \|S^{-1/2}_n\|_2 \|\overline{X}_n\|.
		\end{equation}
		Since the existence of the moment generating function implies that $\mathbb{E}\|X\|^4 < \infty$, Theorem 5.2 of \cite{barndorff1963limit} gives $n^{-1/4}\max_{i=1,\ldots,n}\|X_i\| \rightarrow 0$ almost surely.
		As $S_n - \Sigma = n^{-1}\sum^n_{i=1}(X_i X_i^\top - \textbf{I}_p) - \overline{X}_n \overline{X}^\top_n$, Kolmogorov's variance criterion for averages (see \cite{kallenberg2021foundations} p.113) implies that $n^{1/2-\varepsilon}\|S_n - \Sigma\|_2 \rightarrow 0$ almost surely for any $\varepsilon > 0$. Lemma \ref{lemma:matrix_sqrt_conv} then yields $n^{1/2-\varepsilon}\|S^{-1/2}_n - \Sigma^{-1/2}\|_2 \rightarrow 0$ almost surely. From the proof of Lemma \ref{lemma:matrix_sqrt_conv}, we also know that $\sup_{n\geq 1}\|S^{-1/2}_n\|_2 < \infty$ almost surely. By the strong law of large numbers, $\|\overline{X}_n\| \rightarrow 0$ almost surely. In view of (\ref{eq:xi_n_ineq}),  with probability one,
		\begin{equation*}
			\xi_n \leq o(n^{-1/4+\varepsilon}) o(1) + o(1) = o(1)
		\end{equation*}
	for $\varepsilon < 1/4$. Thus,
		\begin{equation}\label{eq:xi_n_limit}
			\lim_{n \rightarrow \infty} \xi_n  = 0, \quad \text{a.s.}
		\end{equation}
		By Taylor's theorem, for some $|\theta_{n,i}(t)|\leq 1$, we have
		\begin{equation}\label{eq:Taylor_Deltao}
			e^{t^\top \tilde{\Delta}_{n,i}}  = 1 + t^\top \tilde{\Delta}_{n,i} + \frac{1}{2} (t^\top \tilde{\Delta}_{n,i})^2 e^{\theta_{n,i}(t) t^\top \tilde{\Delta}_{n,i}}.
		\end{equation}
		 By (\ref{eq:Taylor_Deltao}),
			\begin{align*}
				M^{(n)}(t) - \tilde{M}^{(n)}(t) 	&=\frac{1}{n}\sum^n_{i=1} e^{t^\top \tilde{X}_i}( e^{t^\top \tilde{\Delta}_{n,i}} - 1)\\
				&=\frac{1}{n} \sum^n_{i=1} e^{t^\top \tilde{X}_i} \left\{ t^\top \tilde{\Delta}_{n,i} + \frac{1}{2}(t^\top \tilde{\Delta}_{n,i})^2 e^{\theta_{n,i}(t) t^\top \tilde{\Delta}_{n,i}} \right\}.
			\end{align*}
			Thus, as $\xi_n = \max_{i=1,\ldots,n} \|\tilde{\Delta}_{n,i}\|$,
			\begin{align}
				|M^{(n)}(t)  - \tilde{M}^{(n)}(t)| & \leq \left| \tilde{M}^{(n)}(t) \right| \max_{i=1,\ldots,n} \left| 
				t^\top \tilde{\Delta}_{n,i} + \frac{1}{2}(t^\top \tilde{\Delta}_{n,i})^2 e^{\theta_{n,i}(t) t^\top \tilde{\Delta}_{n,i}} \right| \nonumber \\
				& \leq \left|\tilde{M}^{(n)}(t)\right| \max_{i=1,\ldots,n}  \bigg(\|t\| \cdot \|\tilde{\Delta}_{n,i}\|
				+ \frac{1}{2} \|t\|^2 \cdot \| \tilde{\Delta}_{n,i}\|^2 e^{\theta_{n,i}(t) \|t\| \cdot \| \tilde{\Delta}_{n,i}\|}\bigg) \nonumber  \\
				& \leq \left|\tilde{M}^{(n)}(t) \right| \bigg(\|t\| \xi_n
				+ \frac{1}{2} \|t\|^2 \xi_n^2 e^{ \|t\| \xi_n}\bigg). \label{eq:consistency_lemma_M}
			\end{align}
Recall that $\tilde{M}(t) := \mathbb{E}(e^{t^\top \tilde{X}})$.
			By the strong law of large numbers, $\tilde{M}^{(n)}(t) \stackrel{a.s.}{\rightarrow} \tilde{M}(t)$. Combining (\ref{eq:consistency_lemma_M}) and (\ref{eq:xi_n_limit}), we deduce that
			\begin{equation}
				M^{(n)}(t) - \tilde{M}^{(n)}(t) \stackrel{a.s.}{\rightarrow}  0.
			\end{equation}
Similarly, using (\ref{eq:Taylor_Deltao}), we have
			\begin{align*}
				& \triangledown M^{(n)}(t) - \triangledown \tilde{M}^{(n)}(t) \\
				&= \frac{1}{n}\sum^n_{i=1}(Z_{n,i} - \tilde{X}_i) e^{t^\top Z_{n,i}}+ \frac{1}{n}\sum^n_{i=1}\tilde{X}_i (e^{t^\top Z_{n,i}} - e^{t^\top \tilde{X}_i} ) \\
				&= \frac{1}{n}\sum^n_{i=1} e^{t^\top \tilde{X}_i} \tilde{\Delta}_{n,i}\cdot e^{t^\top \tilde{\Delta}_{n,i}} + 
				\frac{1}{n}\sum^n_{i=1} \tilde{X}_i e^{t^\top \tilde{X}_i} \left\{ t^\top \tilde{\Delta}_{n,i} + \frac{1}{2} (t^\top \tilde{\Delta}_{n,i})^2 e^{\theta_{n,i}(t) t^\top \tilde{\Delta}_{n,i}}\right\}.
			\end{align*}
By the facts that $\tilde{M}^{(n)}(t) \stackrel{a.s.}{\rightarrow} \tilde{M}(t)$, $\|\triangledown \tilde{M}^{(n)}(t)\| \stackrel{a.s.}{\rightarrow} \| \triangledown \tilde{M}(t)\|$, and (\ref{eq:xi_n_limit}), we obtain that
			\begin{align*}
				&\left\| \triangledown M^{(n)}(t) - \triangledown \tilde{M}^{(n)}(t) \right\| \\
				& \leq |\tilde{M}^{(n)}(t)| \max_{i=1,\ldots,n} \|\tilde{\Delta}_{n,i} e^{t^\top \tilde{\Delta}_{n,i}}\| + \|\triangledown \tilde{M}_{n}(t)\| \max_{i=1,\ldots,n} \left| t^\top \tilde{\Delta}_{n,i} + \frac{1}{2} (t^\top \tilde{\Delta}_{n,i})^2 e^{\theta_{n,i}(t) t^\top \tilde{\Delta}_{n,i}}\right|\\
				& \leq |\tilde{M}^{(n)}(t)| \xi_n e^{\|t\| \xi_n} + \|\triangledown \tilde{M}_{n}(t)\| \left(\|t\| \xi_n + \frac{1}{2} \|t\|^2\xi_n^2 e^{\|t\| \xi_n} \right) \stackrel{a.s}{\rightarrow} 0.
			\end{align*}
			Similar to the proof of Lemma \ref{lemma:Mn_Mn0},  we expand
			\begin{align*}
				H_{M^{(n)}}(t) - H_{\tilde{M}^{(n)}}(t) &=: J_{1n} + J_{2n} + J_{3n} + J_{4n},
			\end{align*}
			where
			\begin{align*}
				J_{1n} &:= \frac{1}{n}\sum^n_{i=1} \tilde{\Delta}_{n,i}\tilde{\Delta}_{n,i}^\top e^{t^\top \tilde{X}_i} e^{t^\top \tilde{\Delta}_{n,i}};\\
				J_{2n} &:= \frac{1}{n}\sum^n_{i=1} \tilde{\Delta}_{n,i}\tilde{X}_i^\top e^{t^\top \tilde{X}_i} e^{t^\top \tilde{\Delta}_{n,i}};\\
				J_{3n} &:= \frac{1}{n}\sum^n_{i=1} \tilde{X}_i \tilde{\Delta}_{n,i}^\top e^{t^\top \tilde{X}_i} e^{t^\top \tilde{\Delta}_{n,i}};\\
				J_{4n} &:= \frac{1}{n} \sum^n_{i=1} \tilde{X}_i \tilde{X}_i^\top e^{t^\top \tilde{X}_i} \left\{ t^\top \tilde{\Delta}_{n,i} + \frac{1}{2}(t^\top \tilde{\Delta}_{n,i})^2e^{\theta_{n,i}(t) t^\top \tilde{\Delta}_{n,i}} \right\}.
			\end{align*}
We  shall show that the spectral norms of each of the above terms all converge	to $0$ almost surely.		For $J_{1,n}$, we see that
			\begin{equation*}
				\|J_{1n}\|_2  \leq  \xi^2_n |\tilde{M}^{(n)}(t)| e^{\|t\|\xi_n}\stackrel{a.s.}{\rightarrow} 0.
			\end{equation*}
			For $J_{2n}$, we have
			\begin{equation*}
				\|J_{2n}\|_2 \leq \frac{1}{n}\sum^n_{i=1}\|\tilde{\Delta}_{n,i}\| \cdot \|\tilde{X}_i\| e^{t^\top \tilde{X}_i} e^{\|t\| \xi_n} \leq \xi_n e^{\|t\| \xi_n} \left( \frac{1}{n}\sum^n_{i=1}\|\tilde{X}_i\|e^{t^\top \tilde{X}_i}\right) \stackrel{a.s.}{\rightarrow} 0,
			\end{equation*}
		as $n^{-1}\sum^n_{i=1}\|\tilde{X}_i\| e^{t^\top \tilde{X}_i} \stackrel{a.s.}{\rightarrow} \mathbb{E}(\|\tilde{X}\|e^{t^\top \tilde{X}})$.
			As $J_{3n} = J_{2n}^\top$, we have $\|J_{3n}\|_2 \stackrel{a.s.}{\rightarrow} 0$. Finally, by the fact that $H_{\tilde{M}^{(n)}}(t) \stackrel{a.s.}{\rightarrow} H_{\tilde{M}}(t)$ and (\ref{eq:xi_n_limit}),
			\begin{equation*}
				\|J_{4n}\|_2 \leq \|H_{\tilde{M}^{(n)}}(t)\|_2 \left(\|t\| \xi_n + \frac{1}{2} \|t\|^2\xi_n^2 e^{\|t\| \xi_n} \right) \stackrel{a.s.}{\rightarrow} 0.
			\end{equation*}

	\end{proof}

	\begin{proof}[Proof of Theorem \ref{thm:consistency_limit}]
We only prove the multivariate case as the univariate case follows by changing the notation.	Simple algebra shows that
		\begin{align*}
		H_{\Lambda^{(n)}}(t) - \textbf{I}_p =	\frac{M^{(n)}(t) H_{M^{(n)}}(t) - \triangledown M^{(n)}(t) (\triangledown M^{(n)}(t))^\top}{(M^{(n)}(t))^2} - \textbf{I}_p =: Q_{1n}(t) + Q_{2n}(t) + Q_{3n}(t),
		\end{align*}
		where
		\begin{align*}
			Q_{1n}(t) &:= \frac{\tilde{M}^{(n)}(t) H_{\tilde{M}^{(n)}}(t) - \triangledown \tilde{M}^{(n)}(t) (\triangledown \tilde{M}^{(n)}(t))^\top}{(\tilde{M}^{(n)}(t))^2} - \textbf{I}_p;\\
			Q_{2n}(t) &:= \left\{M^{(n)}(t) H_{M^{(n)}}(t) - \triangledown M^{(n)}(t) (\triangledown M^{(n)}(t))^\top \right\}  \left\{ \frac{ (\tilde{M}^{(n)}(t))^2 - (M^{(n)}(t))^2 }{ (M^{(n)}(t) \tilde{M}^{(n)}(t))^2 } \right\}; \\
			Q_{3n}(t) &:=\frac{1}{(\tilde{M}^{(n)}(t))^2} 	\bigg\{M^{(n)}(t) H_{M^{(n)}}(t) - \triangledown M^{(n)}(t) (\triangledown M^{(n)}(t))^\top  \\
			& \quad \quad\quad- \tilde{M}^{(n)}(t) H_{\tilde{M}^{(n)}}(t) + \triangledown \tilde{M}^{(n)}(t) (\triangledown \tilde{M}^{(n)}(t))^\top \bigg\}.
		\end{align*}
		By the strong law of large numbers, 
		\begin{equation*}
			Q_{1n}(t) \stackrel{a.s.}{\rightarrow} \frac{\tilde{M}(t)H_{\tilde{M}}(t) - \triangledown \tilde{M}(t) (\triangledown \tilde{M}(t))^\top}{(\tilde{M}(t))^2} - \textbf{I}_p = H_{\tilde{\Lambda}}(t) - \textbf{I}_p.
		\end{equation*}
	According to Lemma \ref{lemma:alternative_M}, it is straightforward to show that $Q_{2n}(t) \stackrel{a.s.}{\rightarrow} 0$ and $Q_{3n}(t) \stackrel{a.s.}{\rightarrow} 0$. The claims in the theorem then follow from the continuous mapping theorem.
	\end{proof}

\setlength{\tabcolsep}{6pt} 
\renewcommand{\arraystretch}{0.5} 
\subsection{Details of Simulation Results for Univariate Test}
\begin{table}[H]
	\centering
	\begin{tabular}{lrrrrrrrrrrr}
		\hline
		$R$	& 0.5 & 1 & 2 & 3 & 4 & 5 & 6 & 7 & 8 & 9 & 10 \\ 
		\hline
		$N(0, 1)$ & 5 & 6 & 5 & 5 & 5 & 5 & 5 & 5 & 5 & 6 & 5 \\ 
		$N(0, 2)$ & 5 & 6 & 5 & 5 & 5 & 5 & 5 & 5 & 5 & 5 & 5 \\ 
		$N(2, 1)$ & 5 & 5 & 5 & 6 & 5 & 5 & 5 & 5 & 5 & 5 & 5 \\ 
		$N(2, 2)$ & 5 & 5 & 5 & 5 & 5 & 5 & 5 & 5 & 5 & 5 & 5 \\ 
		$U(0, 1)$ & 0 & 0 & 77 & 93 & 92 & 92 & 92 & 91 & 89 & 88 & 87 \\ 
		Beta$(0.5,0.5)$ & 0 & 0 & 100 & 100 & 100 & 100 & 100 & 100 & 100 & 100 & 100 \\ 
		Beta$(2,2)$ & 0 & 0 & 15 & 33 & 29 & 26 & 23 & 21 & 19 & 17 & 16 \\ 
		GLD$(0,1,0.25,0.25)$ & 1 & 0 & 2 & 6 & 5 & 4 & 4 & 4 & 3 & 3 & 3 \\ 
		GLD$(0,1,0.5,0.5)$ & 0 & 0 & 22 & 42 & 38 & 34 & 33 & 29 & 27 & 24 & 22 \\ 
		GLD$(0,1,0.75,0.75)$ & 0 & 0 & 59 & 81 & 80 & 78 & 76 & 73 & 70 & 67 & 66 \\ 
		GLD$(0,1,1.25,1.25)$ & 0 & 0 & 85 & 96 & 96 & 95 & 95 & 94 & 94 & 93 & 93 \\ 
		Trunc$(-2,2,0,1)$ & 0 & 0 & 4 & 11 & 9 & 8 & 8 & 7 & 7 & 6 & 5 \\ 
		Trunc$(-3,3,0,2)$ & 0 & 0 & 19 & 38 & 36 & 34 & 34 & 31 & 29 & 26 & 26 \\ 
		Trunc$(-2,2,0,2)$ & 0 & 0 & 48 & 74 & 73 & 70 & 69 & 67 & 65 & 62 & 60 \\ 
		Laplace & 43 & 47 & 41 & 36 & 36 & 38 & 40 & 41 & 42 & 42 & 41 \\ 
		Logistic & 22 & 24 & 20 & 17 & 18 & 19 & 19 & 20 & 20 & 20 & 19 \\ 
		Cauchy & 98 & 99 & 98 & 98 & 98 & 98 & 98 & 99 & 99 & 99 & 99 \\ 
		GLD$(0,1, -0.1, -0.1)$ & 39 & 43 & 39 & 34 & 35 & 35 & 37 & 37 & 38 & 37 & 38 \\ 
		GLD$(0,1, -0.15, -0.15)$ & 48 & 53 & 48 & 43 & 43 & 44 & 46 & 46 & 47 & 47 & 47 \\ 
		t$(5)$ & 37 & 40 & 37 & 32 & 32 & 33 & 35 & 35 & 36 & 36 & 36 \\ 
		t$(10)$ & 18 & 20 & 18 & 15 & 15 & 16 & 16 & 16 & 16 & 17 & 17 \\ 
		t$(15)$ & 13 & 14 & 12 & 10 & 10 & 10 & 11 & 12 & 12 & 12 & 12 \\ 
		Exp$(1)$ & 98 & 95 & 100 & 100 & 100 & 100 & 100 & 100 & 100 & 100 & 100 \\ 
		LogNormal$(0,0.5)$ & 88 & 81 & 87 & 87 & 85 & 85 & 84 & 83 & 83 & 83 & 82 \\ 
		Gamma$(4,5)$ & 62 & 51 & 60 & 59 & 57 & 57 & 56 & 55 & 54 & 54 & 53 \\ 
		Beta$(2,1)$ & 23 & 9 & 65 & 78 & 77 & 77 & 77 & 75 & 75 & 73 & 74 \\ 
		Beta$(3, 2)$ & 3 & 1 & 13 & 25 & 22 & 20 & 19 & 17 & 16 & 14 & 13 \\ 
		Weibull$(3, 1)$ & 4 & 3 & 4 & 6 & 6 & 5 & 6 & 5 & 5 & 5 & 5 \\ 
		Pareto$(1, 3)$ & 100 & 100 & 100 & 100 & 100 & 100 & 100 & 100 & 100 & 100 & 100 \\ 
		chisq$(4)$ & 87 & 78 & 89 & 90 & 90 & 88 & 89 & 89 & 89 & 88 & 87 \\ 
		chisq$(10)$ & 54 & 45 & 50 & 48 & 47 & 46 & 47 & 46 & 45 & 44 & 44 \\ 
		chisq$(20)$ & 33 & 26 & 27 & 26 & 26 & 25 & 26 & 25 & 25 & 24 & 24 \\ 
		ScConN$(0.2, 5)$ & 90 & 95 & 91 & 88 & 89 & 89 & 91 & 91 & 92 & 92 & 93 \\ 
		ScConN$(0.05, 5)$ & 68 & 70 & 69 & 67 & 67 & 68 & 68 & 67 & 68 & 69 & 69 \\ 
		LoConN$(0.5, 3)$ & 0 & 0 & 31 & 40 & 32 & 27 & 22 & 20 & 18 & 16 & 15 \\ 
		LoConN$(0.5, 2)$ & 1 & 1 & 5 & 9 & 8 & 7 & 6 & 6 & 5 & 5 & 4 \\ 
		\hline
	\end{tabular}
	\caption{Reject proportions of our proposed test $U^{(n)}_N$ in the univariate case when $n=50$ and $N=500$ with different values of $R$}
\end{table}

\begin{table}[H]
	\centering
	\begingroup
	\begin{tabular}{lrrrrrrrrrr}
		\hline
		N & 100 & 200 & 300 & 400 & 500 & 600 & 700 & 800 & 900 & 1000 \\ 
		\hline
		$N(0, 1)$ & 5 & 5 & 5 & 5 & 5 & 5 & 5 & 5 & 5 & 5 \\ 
		$N(0, 2)$ & 5 & 5 & 5 & 5 & 5 & 5 & 5 & 5 & 5 & 5 \\ 
		$N(2, 1)$ & 5 & 5 & 5 & 5 & 5 & 5 & 5 & 5 & 5 & 5 \\ 
		$N(2, 2)$ & 5 & 5 & 5 & 5 & 5 & 5 & 5 & 5 & 5 & 5 \\ 
		$U(0, 1)$ & 93 & 93 & 93 & 93 & 93 & 93 & 93 & 93 & 93 & 93 \\ 
		Beta$(0.5,0.5)$ & 100 & 100 & 100 & 100 & 100 & 100 & 100 & 100 & 100 & 100 \\ 
		Beta$(2,2)$ & 33 & 33 & 33 & 33 & 32 & 32 & 32 & 31 & 31 & 32 \\ 
		GLD$(0,1,0.25,0.25)$ & 7 & 7 & 7 & 6 & 6 & 7 & 6 & 6 & 6 & 6 \\ 
		GLD$(0,1,0.5,0.5)$ & 44 & 42 & 42 & 43 & 42 & 42 & 42 & 41 & 41 & 41 \\ 
		GLD$(0,1,0.75,0.75)$ & 81 & 81 & 81 & 81 & 80 & 81 & 80 & 80 & 80 & 80 \\ 
		GLD$(0,1,1.25,1.25)$ & 96 & 96 & 96 & 96 & 96 & 96 & 96 & 95 & 96 & 96 \\ 
		Trunc$(-2,2,0,1)$ & 12 & 11 & 11 & 11 & 11 & 11 & 10 & 10 & 11 & 10 \\ 
		Trunc$(-3,3,0,2)$ & 39 & 39 & 40 & 38 & 37 & 38 & 38 & 38 & 37 & 39 \\ 
		Trunc$(-2,2,0,2)$ & 74 & 73 & 74 & 74 & 73 & 74 & 73 & 73 & 72 & 73 \\ 
		Laplace & 36 & 36 & 36 & 36 & 36 & 35 & 36 & 35 & 35 & 34 \\ 
		Logistic & 18 & 18 & 17 & 18 & 16 & 17 & 18 & 17 & 17 & 18 \\ 
		Cauchy & 98 & 98 & 98 & 98 & 98 & 98 & 98 & 98 & 98 & 98 \\ 
		GLD$(0,1, -0.1, -0.1)$ & 34 & 34 & 34 & 34 & 34 & 33 & 34 & 33 & 34 & 34 \\ 
		GLD$(0,1, -0.15, -0.15)$ & 43 & 42 & 44 & 42 & 43 & 43 & 42 & 43 & 43 & 43 \\ 
		t$(5)$ & 32 & 33 & 31 & 32 & 32 & 32 & 32 & 32 & 32 & 31 \\ 
		t$(10)$ & 15 & 15 & 15 & 15 & 14 & 15 & 14 & 15 & 15 & 15 \\ 
		t$(15)$ & 11 & 10 & 10 & 10 & 10 & 11 & 10 & 10 & 10 & 10 \\ 
		Exp$(1)$ & 100 & 100 & 100 & 100 & 100 & 100 & 100 & 100 & 100 & 100 \\ 
		LogNormal$(0,0.5)$ & 88 & 87 & 87 & 87 & 87 & 87 & 87 & 87 & 87 & 87 \\ 
		Gamma$(4,5)$ & 60 & 59 & 60 & 60 & 59 & 60 & 59 & 58 & 58 & 59 \\ 
		Beta$(2,1)$ & 79 & 79 & 78 & 78 & 78 & 78 & 78 & 78 & 77 & 78 \\ 
		Beta$(3, 2)$ & 26 & 25 & 25 & 25 & 25 & 25 & 25 & 24 & 24 & 24 \\ 
		Weibull$(3, 1)$ & 7 & 7 & 7 & 7 & 6 & 6 & 6 & 6 & 6 & 6 \\ 
		Pareto$(1, 3)$ & 100 & 100 & 100 & 100 & 100 & 100 & 100 & 100 & 100 & 100 \\ 
		chisq$(4)$ & 91 & 90 & 90 & 91 & 90 & 91 & 90 & 90 & 90 & 90 \\ 
		chisq$(10)$ & 49 & 49 & 48 & 49 & 49 & 49 & 48 & 49 & 49 & 49 \\ 
		chisq$(20)$ & 27 & 26 & 27 & 27 & 27 & 27 & 25 & 26 & 27 & 27 \\ 
		ScConN$(0.2, 5)$ & 88 & 88 & 88 & 88 & 88 & 88 & 88 & 88 & 88 & 87 \\ 
		ScConN$(0.05, 5)$ & 67 & 67 & 67 & 67 & 66 & 67 & 67 & 67 & 67 & 67 \\ 
		LoConN$(0.5, 3)$ & 41 & 41 & 40 & 40 & 40 & 39 & 40 & 40 & 39 & 39 \\ 
		LoConN$(0.5, 2)$ & 10 & 10 & 9 & 10 & 9 & 10 & 9 & 9 & 9 & 9 \\ 
		\hline
	\end{tabular}
	\endgroup
	\caption{Reject proportions of our proposed test $U^{(n)}_N$ in the univariate case when $n=50$ and $R=3$ with different values of $N$}
\end{table}

\begin{table}[H]
	\centering
	\begin{tabular}{lrrrrrrrr}
		\hline
		& $U^{(25)}_{2,500}$ & $U^{(25)}_{3,500}$ & $U^{(25)}_{4,500}$ & CvM & AD & SW & JB & HV \\ 
		\hline
		$N(0, 1)$ & 5 & 5 & 5 & 5 & 5 & 5 & 3 & 5 \\ 
		$N(0, 2)$ & 5 & 5 & 5 & 5 & 5 & 5 & 3 & 5 \\ 
		$N(2, 1)$ & 5 & 5 & 5 & 5 & 5 & 5 & 3 & 5 \\ 
		$N(2, 2)$ & 5 & 5 & 5 & 5 & 5 & 5 & 3 & 5 \\ 
		$U(0, 1)$ & 31 & 47 & 47 & 18 & 23 & 29 & 0 & 0 \\ 
		Beta$(0.5,0.5)$ & 85 & 94 & 95 & 64 & 76 & 86 & 0 & 0 \\ 
		Beta$(2,2)$ & 6 & 12 & 10 & 6 & 7 & 6 & 0 & 0 \\ 
		GLD$(0,1,0.25,0.25)$ & 3 & 4 & 4 & 4 & 4 & 4 & 0 & 1 \\ 
		GLD$(0,1,0.5,0.5)$ & 8 & 14 & 13 & 7 & 8 & 8 & 0 & 0 \\ 
		GLD$(0,1,0.75,0.75)$ & 19 & 32 & 32 & 13 & 16 & 18 & 0 & 0 \\ 
		GLD$(0,1,1.25,1.25)$ & 37 & 54 & 54 & 22 & 27 & 35 & 0 & 0 \\ 
		Trunc$(-2,2,0,1)$ & 2 & 5 & 5 & 4 & 4 & 4 & 0 & 1 \\ 
		Trunc$(-3,3,0,2)$ & 6 & 12 & 12 & 6 & 6 & 6 & 0 & 0 \\ 
		Trunc$(-2,2,0,2)$ & 15 & 27 & 25 & 10 & 12 & 14 & 0 & 0 \\ 
		Laplace & 27 & 24 & 23 & 32 & 32 & 32 & 28 & 32 \\ 
		Logistic & 14 & 13 & 12 & 11 & 11 & 13 & 12 & 16 \\ 
		Cauchy & 87 & 83 & 85 & 94 & 94 & 93 & 90 & 89 \\ 
		GLD$(0,1, -0.1, -0.1)$ & 24 & 21 & 21 & 20 & 22 & 24 & 24 & 27 \\ 
		GLD$(0,1, -0.15, -0.15)$ & 30 & 26 & 26 & 27 & 29 & 31 & 30 & 34 \\ 
		t$(5)$ & 23 & 20 & 19 & 17 & 19 & 22 & 22 & 25 \\ 
		t$(10)$ & 11 & 10 & 10 & 8 & 9 & 11 & 10 & 13 \\ 
		t$(15)$ & 9 & 8 & 8 & 7 & 8 & 9 & 7 & 10 \\ 
		Exp$(1)$ & 88 & 88 & 88 & 84 & 88 & 93 & 63 & 74 \\ 
		LogNormal$(0,0.5)$ & 60 & 59 & 57 & 53 & 57 & 64 & 44 & 54 \\ 
		Gamma$(4,5)$ & 34 & 34 & 33 & 27 & 30 & 36 & 22 & 31 \\ 
		Beta$(2,1)$ & 30 & 38 & 40 & 30 & 34 & 41 & 4 & 8 \\ 
		Beta$(3, 2)$ & 7 & 11 & 10 & 8 & 9 & 9 & 1 & 2 \\ 
		Weibull$(3, 1)$ & 4 & 5 & 5 & 5 & 5 & 5 & 2 & 3 \\ 
		Pareto$(1, 3)$ & 97 & 98 & 98 & 96 & 97 & 98 & 86 & 91 \\ 
		chisq$(4)$ & 60 & 60 & 61 & 52 & 57 & 65 & 39 & 50 \\ 
		chisq$(10)$ & 29 & 28 & 28 & 22 & 25 & 31 & 18 & 26 \\ 
		chisq$(20)$ & 17 & 16 & 16 & 13 & 15 & 17 & 11 & 17 \\ 
		ScConN$(0.2, 5)$ & 68 & 64 & 64 & 77 & 80 & 80 & 75 & 74 \\ 
		ScConN$(0.05, 5)$ & 43 & 42 & 42 & 36 & 39 & 42 & 42 & 44 \\ 
		LoConN$(0.5, 3)$ & 15 & 21 & 18 & 21 & 20 & 17 & 0 & 0 \\ 
		LoConN$(0.5, 2)$ & 4 & 6 & 6 & 6 & 5 & 5 & 1 & 1 \\ 
		\hline
	\end{tabular}
	\caption{Reject proportions of our proposed test $U^{(n)}_{N,R}$ ($R=2, 3, 4$, $N = 500$) and other tests in the univariate case when $n=25$}
\end{table}

\begin{table}[H]
	\centering
	\begin{tabular}{lrrrrrrrr}
		\hline
		& $U^{50}_{2,500}$ & $U^{50}_{3,500}$ & $U^{50}_{4,500}$ & CvM & AD & SW & JB & HV \\ 
		\hline
		$N(0, 1)$ & 5 & 5 & 5 & 5 & 5 & 5 & 4 & 5 \\ 
		$N(0, 2)$ & 5 & 5 & 5 & 6 & 5 & 5 & 4 & 5 \\ 
		$N(2, 1)$ & 5 & 6 & 5 & 5 & 5 & 5 & 4 & 5 \\ 
		$N(2, 2)$ & 5 & 5 & 5 & 5 & 5 & 5 & 4 & 5 \\ 
		$U(0, 1)$ & 78 & 93 & 92 & 44 & 57 & 76 & 0 & 0 \\ 
		Beta$(0.5,0.5)$ & 100 & 100 & 100 & 96 & 99 & 100 & 3 & 0 \\ 
		Beta$(2,2)$ & 16 & 33 & 29 & 12 & 14 & 16 & 0 & 0 \\ 
		GLD$(0,1,0.25,0.25)$ & 3 & 6 & 5 & 5 & 5 & 4 & 0 & 1 \\ 
		GLD$(0,1,0.5,0.5)$ & 23 & 42 & 38 & 14 & 17 & 21 & 0 & 0 \\ 
		GLD$(0,1,0.75,0.75)$ & 59 & 81 & 80 & 31 & 40 & 55 & 0 & 0 \\ 
		GLD$(0,1,1.25,1.25)$ & 85 & 96 & 96 & 52 & 66 & 83 & 0 & 0 \\ 
		Trunc$(-2,2,0,1)$ & 4 & 11 & 9 & 5 & 5 & 5 & 0 & 0 \\ 
		Trunc$(-3,3,0,2)$ & 17 & 38 & 36 & 10 & 12 & 16 & 0 & 0 \\ 
		Trunc$(-2,2,0,2)$ & 48 & 74 & 73 & 22 & 31 & 45 & 0 & 0 \\ 
		Laplace & 42 & 36 & 36 & 54 & 55 & 53 & 52 & 47 \\ 
		Logistic & 21 & 17 & 18 & 14 & 16 & 20 & 23 & 23 \\ 
		Cauchy & 99 & 98 & 98 & 100 & 100 & 100 & 99 & 99 \\ 
		GLD$(0,1, -0.1, -0.1)$ & 39 & 34 & 35 & 32 & 35 & 40 & 44 & 42 \\ 
		GLD$(0,1, -0.15, -0.15)$ & 48 & 43 & 43 & 44 & 48 & 51 & 54 & 52 \\ 
		t$(5)$ & 37 & 32 & 32 & 26 & 30 & 35 & 40 & 39 \\ 
		t$(10)$ & 18 & 15 & 15 & 11 & 12 & 16 & 18 & 19 \\ 
		t$(15)$ & 12 & 10 & 10 & 8 & 8 & 10 & 12 & 13 \\ 
		Exp$(1)$ & 100 & 100 & 100 & 99 & 100 & 100 & 96 & 96 \\ 
		LogNormal$(0,0.5)$ & 87 & 87 & 85 & 83 & 87 & 92 & 79 & 84 \\ 
		Gamma$(4,5)$ & 60 & 59 & 57 & 52 & 59 & 69 & 50 & 57 \\ 
		Beta$(2,1)$ & 66 & 78 & 77 & 62 & 72 & 83 & 10 & 14 \\ 
		Beta$(3, 2)$ & 13 & 25 & 22 & 15 & 17 & 20 & 1 & 2 \\ 
		Weibull$(3, 1)$ & 5 & 6 & 6 & 6 & 6 & 6 & 2 & 4 \\ 
		Pareto$(1, 3)$ & 100 & 100 & 100 & 100 & 100 & 100 & 100 & 100 \\ 
		chisq$(4)$ & 90 & 90 & 90 & 84 & 89 & 95 & 77 & 82 \\ 
		chisq$(10)$ & 49 & 48 & 47 & 42 & 48 & 58 & 42 & 49 \\ 
		chisq$(20)$ & 28 & 26 & 26 & 23 & 26 & 33 & 24 & 30 \\ 
		ScConN$(0.2, 5)$ & 91 & 88 & 89 & 95 & 97 & 97 & 97 & 93 \\ 
		ScConN$(0.05, 5)$ & 69 & 67 & 67 & 55 & 59 & 66 & 68 & 69 \\ 
		LoConN$(0.5, 3)$ & 30 & 40 & 32 & 44 & 44 & 37 & 0 & 0 \\ 
		LoConN$(0.5, 2)$ & 5 & 9 & 8 & 7 & 8 & 6 & 0 & 1 \\ 
		\hline
	\end{tabular}
	\caption{Reject proportions of our proposed test $U^{(n)}_{R,N}$ ($R=2,3,4$, $N=500$) and other tests in the univariate case when $n=50$}
\end{table}

\begin{table}[H]
	\centering
	\begin{tabular}{lrrrrrrrr}
		\hline
		& $U^{(100)}_{2,500}$ & $U^{(100)}_{3,500}$ & $U^{(100)}_{4,500}$ & CvM & AD & SW & JB & HV \\ 
		\hline
		$N(0, 1)$ & 5 & 5 & 5 & 5 & 5 & 5 & 4 & 5 \\ 
		$N(0, 2)$ & 5 & 5 & 5 & 5 & 5 & 5 & 4 & 5 \\ 
		$N(2, 1)$ & 5 & 5 & 5 & 5 & 5 & 5 & 4 & 5 \\ 
		$N(2, 2)$ & 5 & 5 & 5 & 5 & 5 & 5 & 4 & 5 \\ 
		$U(0, 1)$ & 99 & 100 & 100 & 84 & 95 & 100 & 56 & 0 \\ 
		Beta$(0.5,0.5)$ & 100 & 100 & 100 & 100 & 100 & 100 & 100 & 0 \\ 
		Beta$(2,2)$ & 42 & 77 & 73 & 25 & 31 & 45 & 2 & 0 \\ 
		GLD$(0,1,0.25,0.25)$ & 3 & 11 & 10 & 6 & 6 & 6 & 0 & 0 \\ 
		GLD$(0,1,0.5,0.5)$ & 58 & 88 & 86 & 31 & 42 & 60 & 3 & 0 \\ 
		GLD$(0,1,0.75,0.75)$ & 96 & 100 & 100 & 66 & 83 & 96 & 30 & 0 \\ 
		GLD$(0,1,1.25,1.25)$ & 100 & 100 & 100 & 90 & 98 & 100 & 69 & 0 \\ 
		Trunc$(-2,2,0,1)$ & 8 & 35 & 33 & 7 & 9 & 13 & 0 & 0 \\ 
		Trunc$(-3,3,0,2)$ & 51 & 88 & 88 & 22 & 32 & 55 & 2 & 0 \\ 
		Trunc$(-2,2,0,2)$ & 91 & 100 & 100 & 53 & 71 & 92 & 17 & 0 \\ 
		Laplace & 61 & 54 & 55 & 82 & 83 & 80 & 79 & 66 \\ 
		Logistic & 31 & 25 & 26 & 21 & 24 & 31 & 37 & 33 \\ 
		Cauchy & 100 & 100 & 100 & 100 & 100 & 100 & 100 & 100 \\ 
		GLD$(0,1, -0.1, -0.1)$ & 58 & 50 & 52 & 53 & 57 & 62 & 68 & 60 \\ 
		GLD$(0,1, -0.15, -0.15)$ & 69 & 63 & 65 & 69 & 72 & 76 & 79 & 71 \\ 
		t$(5)$ & 56 & 50 & 51 & 43 & 48 & 57 & 63 & 57 \\ 
		t$(10)$ & 27 & 22 & 22 & 14 & 17 & 24 & 29 & 28 \\ 
		t$(15)$ & 17 & 14 & 14 & 9 & 10 & 14 & 17 & 18 \\ 
		Exp$(1)$ & 100 & 100 & 100 & 100 & 100 & 100 & 100 & 100 \\ 
		LogNormal$(0,0.5)$ & 99 & 99 & 99 & 99 & 99 & 100 & 99 & 99 \\ 
		Gamma$(4,5)$ & 86 & 88 & 86 & 84 & 89 & 96 & 86 & 88 \\ 
		Beta$(2,1)$ & 96 & 99 & 99 & 94 & 98 & 100 & 74 & 42 \\ 
		Beta$(3, 2)$ & 29 & 60 & 56 & 32 & 39 & 53 & 5 & 3 \\ 
		Weibull$(3, 1)$ & 5 & 10 & 9 & 7 & 7 & 8 & 3 & 5 \\ 
		Pareto$(1, 3)$ & 100 & 100 & 100 & 100 & 100 & 100 & 100 & 100 \\ 
		chisq$(4)$ & 100 & 100 & 100 & 99 & 100 & 100 & 99 & 99 \\ 
		chisq$(10)$ & 75 & 77 & 74 & 73 & 80 & 90 & 78 & 81 \\ 
		chisq$(20)$ & 43 & 43 & 41 & 42 & 48 & 60 & 48 & 54 \\ 
		ScConN$(0.2, 5)$ & 99 & 99 & 99 & 100 & 100 & 100 & 100 & 100 \\ 
		ScConN$(0.05, 5)$ & 89 & 88 & 88 & 75 & 80 & 88 & 89 & 89 \\ 
		LoConN$(0.5, 3)$ & 53 & 64 & 52 & 81 & 82 & 75 & 12 & 0 \\ 
		LoConN$(0.5, 2)$ & 6 & 16 & 12 & 13 & 13 & 11 & 0 & 1 \\ 
		\hline
	\end{tabular}
	\caption{Reject proportions of our proposed test $U^{(n)}_{R,N}$ ($R=2,3,4$, $N=500$) and other tests in the univariate case when $n=100$}
\end{table}

\begin{table}[H]
	\centering
	\begin{tabular}{lrrrrrrrr}
		\hline
		& $U^{(200)}_{2,500}$ & $U^{(200)}_{3,500}$ & $U^{(200)}_{4,500}$ & CvM & AD & SW & JB & HV \\ 
		\hline
		$N(0, 1)$ & 5 & 5 & 5 & 5 & 5 & 5 & 5 & 5 \\ 
		$N(0, 2)$ & 5 & 5 & 5 & 5 & 5 & 5 & 5 & 5 \\ 
		$N(2, 1)$ & 5 & 5 & 5 & 5 & 5 & 5 & 5 & 5 \\ 
		$N(2, 2)$ & 5 & 5 & 5 & 5 & 5 & 5 & 5 & 5 \\ 
		$U(0, 1)$ & 100 & 100 & 100 & 100 & 100 & 100 & 100 & 0 \\ 
		Beta$(0.5,0.5)$ & 100 & 100 & 100 & 100 & 100 & 100 & 100 & 1 \\ 
		Beta$(2,2)$ & 82 & 100 & 100 & 55 & 71 & 92 & 62 & 0 \\ 
		GLD$(0,1,0.25,0.25)$ & 3 & 24 & 20 & 10 & 11 & 12 & 2 & 0 \\ 
		GLD$(0,1,0.5,0.5)$ & 94 & 100 & 100 & 69 & 85 & 98 & 80 & 0 \\ 
		GLD$(0,1,0.75,0.75)$ & 100 & 100 & 100 & 97 & 100 & 100 & 100 & 0 \\ 
		GLD$(0,1,1.25,1.25)$ & 100 & 100 & 100 & 100 & 100 & 100 & 100 & 0 \\ 
		Trunc$(-2,2,0,1)$ & 20 & 85 & 87 & 14 & 19 & 44 & 8 & 0 \\ 
		Trunc$(-3,3,0,2)$ & 91 & 100 & 100 & 52 & 74 & 98 & 67 & 0 \\ 
		Trunc$(-2,2,0,2)$ & 100 & 100 & 100 & 91 & 99 & 100 & 99 & 0 \\ 
		Laplace & 81 & 74 & 76 & 98 & 98 & 98 & 96 & 86 \\ 
		Logistic & 45 & 37 & 38 & 34 & 39 & 49 & 58 & 47 \\ 
		Cauchy & 100 & 100 & 100 & 100 & 100 & 100 & 100 & 100 \\ 
		GLD$(0,1, -0.1, -0.1)$ & 79 & 73 & 73 & 80 & 83 & 87 & 90 & 81 \\ 
		GLD$(0,1, -0.15, -0.15)$ & 89 & 85 & 85 & 92 & 94 & 95 & 96 & 91 \\ 
		t$(5)$ & 76 & 69 & 71 & 68 & 74 & 81 & 86 & 77 \\ 
		t$(10)$ & 38 & 31 & 32 & 20 & 24 & 35 & 44 & 37 \\ 
		t$(15)$ & 24 & 19 & 19 & 12 & 13 & 21 & 27 & 24 \\ 
		Exp$(1)$ & 100 & 100 & 100 & 100 & 100 & 100 & 100 & 100 \\ 
		LogNormal$(0,0.5)$ & 100 & 100 & 100 & 100 & 100 & 100 & 100 & 100 \\ 
		Gamma$(4,5)$ & 99 & 100 & 99 & 99 & 100 & 100 & 100 & 100 \\ 
		Beta$(2,1)$ & 100 & 100 & 100 & 100 & 100 & 100 & 100 & 94 \\ 
		Beta$(3, 2)$ & 62 & 95 & 92 & 66 & 80 & 95 & 65 & 12 \\ 
		Weibull$(3, 1)$ & 5 & 17 & 15 & 10 & 11 & 16 & 6 & 7 \\ 
		Pareto$(1, 3)$ & 100 & 100 & 100 & 100 & 100 & 100 & 100 & 100 \\ 
		chisq$(4)$ & 100 & 100 & 100 & 100 & 100 & 100 & 100 & 100 \\ 
		chisq$(10)$ & 96 & 97 & 96 & 97 & 99 & 100 & 99 & 99 \\ 
		chisq$(20)$ & 63 & 67 & 64 & 73 & 80 & 90 & 85 & 85 \\ 
		ScConN$(0.2, 5)$ & 100 & 100 & 100 & 100 & 100 & 100 & 100 & 100 \\ 
		ScConN$(0.05, 5)$ & 99 & 98 & 99 & 93 & 95 & 98 & 99 & 99 \\ 
		LoConN$(0.5, 3)$ & 73 & 83 & 70 & 99 & 99 & 98 & 86 & 0 \\ 
		LoConN$(0.5, 2)$ & 6 & 24 & 18 & 25 & 27 & 23 & 6 & 0 \\ 
		\hline
	\end{tabular}
	\caption{Reject proportions of our proposed test $U^{(n)}_{R,N}$ ($R=2,3,4$, $N=500$) and other tests in the univariate case when $n=200$}
\end{table}

\subsection{Details Simulation Results for Multivariate Test}
In Table \ref{table:multi_diff_R}, the empirical reject proportions of our test with different values of $R$ are shown when $n = 50, p = 3, N = 500$. The results show that $R = 3$ tend to perform well in different cases. In Tables \ref{table:multi_diff_N_p3} and \ref{table:multi_diff_N_p10}, the empirical reject proportions of our test with different values of $N$ are shown when $n = 50, p = 3,R = 3$ and $n = 50, p = 10, R = 3$, respectively.
In Tables \ref{table:n25p2} to \ref{table:n100p10}, we present the empirical reject proportions of our test $T^{(n)}_{N}$ together with the energy test of \cite{szekely2005new}, the Henze-Visagie (HV) test (defined in (\ref{eq:HV})), the Henze–Jim\'enez-Gamero (HJ) test (\cite{henze2019new}, the Henze-Zirkler (HZ) test (\cite{henze1990class}) and the Mardia's test (\cite{mardia1970measures}) based on skewness (MS) and kurtosis (MK) for different sample sizes $n$ and dimensions $p$. The settings include $n = 25, 50, 100$ and $p = 2, 3, 5, 10$, where $N = 500$ and $R = 3$. We also include the corresponding tests using only $H^{(n)}_N$ and $D^{(n)}_N$ defined in (\ref{eq:Hn_Dn_t}). For normal distributions, $m$ denotes the vector $(1,\ldots,p)$ and $S$ denotes the matrix with diagonal elements being $1$ and off-diagonal elements being $0.5$.

\begin{table}[H]
\centering
\begingroup\footnotesize
\begin{tabular}{lrrrrrrrrrrr}
	\hline
	& 0.5 & 1 & 2 & 3 & 4 & 5 & 6 & 7 & 8 & 9 & 10 \\ 
	\hline
	$N(0, I)$ & 5 & 5 & 5 & 5 & 5 & 5 & 5 & 5 & 5 & 5 & 5 \\ 
	$N(0, S)$ & 5 & 5 & 5 & 5 & 5 & 5 & 5 & 5 & 5 & 5 & 5 \\ 
	$N(m, I)$ & 5 & 5 & 5 & 5 & 5 & 5 & 5 & 5 & 5 & 5 & 5 \\ 
	$N(m, S)$ & 5 & 5 & 5 & 5 & 5 & 5 & 5 & 5 & 5 & 5 & 5 \\ 
	$U(0, 1)^p$ & 0 & 0 & 82 & 100 & 100 & 99 & 99 & 98 & 98 & 97 & 96 \\ 
	Beta$(0.5, 0.5)^p$ & 0 & 0 & 100 & 100 & 100 & 100 & 100 & 100 & 100 & 100 & 100 \\ 
	Beta$(2, 2)^p$ & 0 & 0 & 3 & 36 & 35 & 31 & 28 & 25 & 21 & 17 & 14 \\
\hline
\end{tabular}
\endgroup
\caption{$n = 50, p = 3, N = 500$} 
\label{table:multi_diff_R}
\end{table}

\begin{table}[H]
	\centering
	\begingroup\footnotesize
	\begin{tabular}{lrrrrrrrrrr}
		\hline
		$N$ & 100 & 200 & 300 & 400 & 500 & 600 & 700 & 800 & 900 & 1000 \\ 
		\hline
		$N(0, I)$ & 5 & 5 & 5 & 5 & 5 & 5 & 5 & 5 & 5 & 5 \\ 
		$N(0, S)$ & 5 & 5 & 5 & 5 & 5 & 5 & 5 & 5 & 5 & 5 \\ 
		$N(m, I)$ & 5 & 5 & 5 & 5 & 5 & 5 & 5 & 5 & 5 & 5 \\ 
		$N(m, S)$ & 5 & 5 & 5 & 5 & 5 & 5 & 5 & 5 & 5 & 5 \\ 
		$U(0, 1)^p$ & 99 & 99 & 99 & 99 & 100 & 100 & 100 & 100 & 100 & 100 \\ 
		Beta$(0.5, 0.5)^p$ & 100 & 100 & 100 & 100 & 100 & 100 & 100 & 100 & 100 & 100 \\ 
		Beta$(2, 2)^p$ & 31 & 31 & 32 & 34 & 36 & 37 & 39 & 39 & 39 & 39 \\
\hline
	\end{tabular}
	\endgroup
	\caption{$n = 50, p = 3, R = 3$ } 
\label{table:multi_diff_N_p3}
\end{table}

\begin{table}[H]
	\centering
	\begingroup\footnotesize
	\begin{tabular}{lrrrrrrrrrr}
		\hline
		& 100 & 200 & 300 & 400 & 500 & 600 & 700 & 800 & 900 & 1000 \\ 
		\hline
		$N(0, I)$ & 5 & 5 & 5 & 5 & 5 & 5 & 5 & 5 & 5 & 5 \\ 
		$N(0, S)$ & 5 & 5 & 5 & 5 & 5 & 5 & 5 & 5 & 5 & 5 \\ 
		$N(m, I)$ & 5 & 5 & 5 & 5 & 5 & 5 & 5 & 5 & 5 & 5 \\ 
		$N(m, S)$ & 5 & 5 & 5 & 5 & 5 & 5 & 5 & 5 & 5 & 5 \\ 
		$U(0, 1)^p$ & 100 & 100 & 100 & 100 & 100 & 100 & 100 & 100 & 100 & 100 \\ 
		Beta$(0.5, 0.5)^p$ & 100 & 100 & 100 & 100 & 100 & 100 & 100 & 100 & 100 & 100 \\ 
		Beta$(2, 2)^p$ & 83 & 87 & 88 & 88 & 88 & 87 & 87 & 87 & 87 & 88 \\
\hline
	\end{tabular}
	\endgroup
	\caption{$n = 50, p = 10, R = 3$ }
	\label{table:multi_diff_N_p10}
\end{table}

\begin{table}[H]
	\centering
	\begingroup\footnotesize
	\begin{tabular}{lrrrrrrrrr}
		\hline
		& $H_{n}$ & $D_{n}$ & $T_{n}$ & EN & HV & HZ & HJG & MS & MK \\ 
		\hline
		$N(0, I)$ & 5 & 5 & 5 & 5 & 5 & 5 & 5 & 5 & 5 \\ 
		$N(0, S)$ & 5 & 5 & 5 & 5 & 5 & 5 & 5 & 5 & 4 \\ 
		$N(m, I)$ & 5 & 5 & 5 & 5 & 5 & 5 & 5 & 5 & 5 \\ 
		$N(m, S)$ & 5 & 5 & 5 & 5 & 5 & 5 & 5 & 5 & 5 \\ 
		$U(0, 1)^p$ & 0 & 72 & 40 & 16 & 0 & 25 & 0 & 0 & 49 \\ 
		Beta$(0.5, 0.5)^p$ & 0 & 99 & 95 & 65 & 0 & 77 & 0 & 0 & 82 \\ 
		Beta$(2, 2)^p$ & 0 & 18 & 4 & 5 & 0 & 8 & 0 & 0 & 19 \\
\hline
	\end{tabular}
	\endgroup
	\caption{$n = 25, p = 2$} 
	\label{table:n25p2}
\end{table}

\begin{table}[H]
	\centering
	\begingroup\footnotesize
	\begin{tabular}{lrrrrrrrrr}
		\hline
		& $H_{n}$ & $D_{n}$ & $T_{n}$ & EN & HV & HZ & HJG & MS & MK \\ 
		\hline
		$N(0, I)$ & 6 & 5 & 5 & 5 & 5 & 5 & 5 & 5 & 5 \\ 
		$N(0, S)$ & 5 & 5 & 5 & 5 & 5 & 5 & 5 & 5 & 5 \\ 
		$N(m, I)$ & 5 & 5 & 5 & 5 & 5 & 5 & 5 & 5 & 5 \\ 
		$N(m, S)$ & 5 & 5 & 5 & 5 & 5 & 5 & 5 & 5 & 5 \\ 
		$U(0, 1)^p$ & 0 & 83 & 58 & 9 & 0 & 22 & 0 & 0 & 46 \\ 
		Beta$(0.5, 0.5)^p$ & 0 & 100 & 98 & 42 & 0 & 71 & 0 & 0 & 79 \\ 
		Beta$(2, 2)^p$ & 0 & 23 & 7 & 3 & 1 & 7 & 0 & 1 & 19 \\
\hline
	\end{tabular}
	\endgroup
	\caption{$n = 25, p = 3$} 
\end{table}

\begin{table}[H]
	\centering
	\begingroup\footnotesize
	\begin{tabular}{lrrrrrrrrr}
		\hline
		& $H_{n}$ & $D_{n}$ & $T_{n}$ & EN & HV & HZ & HJG & MS & MK \\ 
		\hline
		$N(0, I)$ & 5 & 5 & 5 & 5 & 5 & 5 & 5 & 5 & 5 \\ 
		$N(0, S)$ & 5 & 5 & 5 & 5 & 5 & 5 & 5 & 5 & 5 \\ 
		$N(m, I)$ & 5 & 5 & 5 & 5 & 5 & 5 & 5 & 5 & 5 \\ 
		$N(m, S)$ & 5 & 5 & 5 & 5 & 5 & 5 & 5 & 5 & 5 \\ 
		$U(0, 1)^p$ & 0 & 91 & 77 & 3 & 0 & 16 & 0 & 0 & 39 \\ 
		Beta$(0.5, 0.5)^p$ & 0 & 100 & 100 & 10 & 0 & 47 & 0 & 0 & 67 \\ 
		Beta$(2, 2)^p$ & 0 & 30 & 13 & 2 & 1 & 7 & 1 & 0 & 17 \\
\hline
	\end{tabular}
	\endgroup
	\caption{$n = 25, p = 5$} 
\end{table}

\begin{table}[H]
	\centering
	\begingroup\footnotesize
	\begin{tabular}{lrrrrrrrrr}
		\hline
		& $H_{n}$ & $D_{n}$ & $T_{n}$ & EN & HV & HZ & HJG & MS & MK \\ 
		\hline
		$N(0, I)$ & 5 & 5 & 5 & 5 & 5 & 5 & 5 & 5 & 5 \\ 
		$N(0, S)$ & 5 & 5 & 5 & 5 & 5 & 5 & 5 & 5 & 5 \\ 
		$N(m, I)$ & 5 & 5 & 5 & 5 & 5 & 5 & 5 & 5 & 5 \\ 
		$N(m, S)$ & 5 & 5 & 5 & 5 & 5 & 5 & 5 & 5 & 5 \\ 
		$U(0, 1)^p$ & 0 & 94 & 85 & 1 & 1 & 8 & 1 & 0 & 21 \\ 
		Beta$(0.5, 0.5)^p$ & 0 & 100 & 100 & 1 & 0 & 16 & 0 & 0 & 35 \\ 
		Beta$(2, 2)^p$ & 0 & 37 & 20 & 1 & 1 & 6 & 1 & 1 & 12 \\
\hline
	\end{tabular}
	\endgroup
	\caption{$n = 25, p = 10$} 
\end{table}

\begin{table}[H]
	\centering
	\begingroup\footnotesize
	\begin{tabular}{lrrrrrrrrr}
		\hline
		& $H_{n}$ & $D_{n}$ & $T_{n}$ & EN & HV & HZ & HJG & MS & MK \\ 
		\hline
		$N(0, I)$ & 5 & 5 & 5 & 5 & 5 & 5 & 5 & 5 & 5 \\ 
		$N(0, S)$ & 5 & 5 & 5 & 5 & 5 & 5 & 5 & 5 & 5 \\ 
		$N(m, I)$ & 5 & 5 & 5 & 5 & 5 & 5 & 5 & 5 & 5 \\ 
		$N(m, S)$ & 5 & 5 & 5 & 5 & 5 & 5 & 5 & 5 & 5 \\ 
		$U(0, 1)^p$ & 0 & 100 & 95 & 52 & 0 & 67 & 0 & 0 & 91 \\ 
		Beta$(0.5, 0.5)^p$ & 0 & 100 & 100 & 99 & 0 & 100 & 0 & 0 & 100 \\ 
		Beta$(2, 2)^p$ & 0 & 55 & 19 & 11 & 0 & 18 & 0 & 0 & 46 \\
\hline
	\end{tabular}
	\endgroup
	\caption{$n = 50, p = 2$} 
\end{table}

\begin{table}[H]
	\centering
	\begingroup\footnotesize
	\begin{tabular}{lrrrrrrrrr}
		\hline
		& $H_{n}$ & $D_{n}$ & $T_{n}$ & EN & HV & HZ & HJG & MS & MK \\ 
		\hline
		$N(0, I)$ & 5 & 5 & 5 & 6 & 5 & 6 & 5 & 6 & 5 \\ 
		$N(0, S)$ & 5 & 5 & 5 & 5 & 5 & 5 & 5 & 5 & 5 \\ 
		$N(m, I)$ & 5 & 5 & 5 & 5 & 5 & 5 & 5 & 5 & 5 \\ 
		$N(m, S)$ & 5 & 5 & 5 & 5 & 5 & 5 & 5 & 5 & 4 \\ 
		$U(0, 1)^p$ & 0 & 100 & 100 & 40 & 0 & 65 & 0 & 0 & 91 \\ 
		Beta$(0.5, 0.5)^p$ & 0 & 100 & 100 & 97 & 0 & 100 & 0 & 0 & 100 \\ 
		Beta$(2, 2)^p$ & 0 & 68 & 36 & 7 & 0 & 17 & 0 & 0 & 48 \\
\hline
	\end{tabular}
	\endgroup
	\caption{$n = 50, p = 3$} 
\end{table}

\begin{table}[H]
	\centering
	\begingroup\footnotesize
	\begin{tabular}{lrrrrrrrrr}
		\hline
		& $H_{n}$ & $D_{n}$ & $T_{n}$ & EN & HV & HZ & HJG & MS & MK \\ 
		\hline
		$N(0, I)$ & 5 & 5 & 5 & 5 & 5 & 5 & 5 & 5 & 5 \\ 
		$N(0, S)$ & 5 & 5 & 5 & 5 & 5 & 5 & 5 & 5 & 5 \\ 
		$N(m, I)$ & 5 & 5 & 5 & 5 & 5 & 5 & 5 & 5 & 5 \\ 
		$N(m, S)$ & 5 & 5 & 5 & 4 & 5 & 5 & 5 & 5 & 5 \\ 
		$U(0, 1)^p$ & 0 & 100 & 100 & 12 & 0 & 51 & 0 & 0 & 89 \\ 
		Beta$(0.5, 0.5)^p$ & 0 & 100 & 100 & 63 & 0 & 97 & 0 & 0 & 100 \\ 
		Beta$(2, 2)^p$ & 0 & 85 & 64 & 3 & 0 & 14 & 0 & 0 & 49 \\
\hline
	\end{tabular}
	\endgroup
	\caption{$n = 50, p = 5$} 
\end{table}

\begin{table}[H]
	\centering
	\begingroup\footnotesize
	\begin{tabular}{lrrrrrrrrr}
		\hline
		& $H_{n}$ & $D_{n}$ & $T_{n}$ & EN & HV & HZ & HJG & MS & MK \\ 
		\hline
		$N(0, I)$ & 5 & 5 & 5 & 5 & 5 & 5 & 5 & 5 & 5 \\ 
		$N(0, S)$ & 5 & 5 & 5 & 5 & 5 & 5 & 5 & 5 & 5 \\ 
		$N(m, I)$ & 5 & 5 & 5 & 5 & 5 & 5 & 5 & 5 & 5 \\ 
		$N(m, S)$ & 5 & 5 & 5 & 5 & 5 & 5 & 5 & 5 & 5 \\ 
		$U(0, 1)^p$ & 0 & 100 & 100 & 1 & 0 & 27 & 0 & 0 & 78 \\ 
		Beta$(0.5, 0.5)^p$ & 0 & 100 & 100 & 2 & 0 & 64 & 0 & 0 & 98 \\ 
		Beta$(2, 2)^p$ & 0 & 96 & 88 & 1 & 0 & 11 & 0 & 0 & 40 \\
\hline
	\end{tabular}
	\endgroup
	\caption{$n = 50, p = 10$} 
\end{table}

\begin{table}[H]
	\centering
	\begingroup\footnotesize
	\begin{tabular}{lrrrrrrrrr}
		\hline
		& $H_{n}$ & $D_{n}$ & $T_{n}$ & EN & HV & HZ & HJG & MS & MK \\ 
		\hline
		$N(0, I)$ & 5 & 5 & 5 & 5 & 5 & 5 & 5 & 5 & 5 \\ 
		$N(0, S)$ & 5 & 5 & 5 & 5 & 5 & 5 & 5 & 5 & 5 \\ 
		$N(m, I)$ & 6 & 5 & 5 & 5 & 5 & 5 & 5 & 5 & 5 \\ 
		$N(m, S)$ & 5 & 5 & 5 & 5 & 5 & 6 & 5 & 5 & 5 \\ 
		$U(0, 1)^p$ & 0 & 100 & 100 & 96 & 0 & 97 & 0 & 0 & 100 \\ 
		Beta$(0.5, 0.5)^p$ & 0 & 100 & 100 & 100 & 0 & 100 & 0 & 0 & 100 \\ 
		Beta$(2, 2)^p$ & 0 & 96 & 74 & 32 & 0 & 42 & 0 & 0 & 86 \\
\hline
	\end{tabular}
	\endgroup
	\caption{$n = 100, p = 2$} 
\end{table}

\begin{table}[H]
	\centering
	\begingroup\footnotesize
	\begin{tabular}{lrrrrrrrrr}
		\hline
		& $H_{n}$ & $D_{n}$ & $T_{n}$ & EN & HV & HZ & HJG & MS & MK \\ 
		\hline
		$N(0, I)$ & 5 & 5 & 5 & 5 & 5 & 5 & 5 & 5 & 5 \\ 
		$N(0, S)$ & 5 & 5 & 5 & 5 & 5 & 5 & 5 & 5 & 5 \\ 
		$N(m, I)$ & 5 & 5 & 5 & 5 & 5 & 5 & 5 & 5 & 5 \\ 
		$N(m, S)$ & 5 & 5 & 5 & 5 & 4 & 5 & 4 & 5 & 5 \\ 
		$U(0, 1)^p$ & 0 & 100 & 100 & 94 & 0 & 98 & 0 & 0 & 100 \\ 
		Beta$(0.5, 0.5)^p$ & 0 & 100 & 100 & 100 & 0 & 100 & 0 & 0 & 100 \\ 
		Beta$(2, 2)^p$ & 0 & 99 & 94 & 25 & 0 & 40 & 0 & 0 & 88 \\
\hline
	\end{tabular}
	\endgroup
	\caption{$n = 100, p = 3$} 
\end{table}

\begin{table}[H]
	\centering
	\begingroup\footnotesize
	\begin{tabular}{lrrrrrrrrr}
		\hline
		& $H_{n}$ & $D_{n}$ & $T_{n}$ & EN & HV & HZ & HJG & MS & MK \\ 
		\hline
		$N(0, I)$ & 5 & 5 & 5 & 5 & 5 & 5 & 5 & 5 & 5 \\ 
		$N(0, S)$ & 5 & 5 & 5 & 5 & 5 & 5 & 5 & 5 & 5 \\ 
		$N(m, I)$ & 5 & 5 & 5 & 5 & 5 & 5 & 5 & 5 & 5 \\ 
		$N(m, S)$ & 5 & 5 & 5 & 5 & 5 & 5 & 5 & 5 & 5 \\ 
		$U(0, 1)^p$ & 0 & 100 & 100 & 75 & 0 & 95 & 0 & 0 & 100 \\ 
		Beta$(0.5, 0.5)^p$ & 0 & 100 & 100 & 100 & 0 & 100 & 0 & 0 & 100 \\ 
		Beta$(2, 2)^p$ & 0 & 100 & 100 & 14 & 0 & 36 & 0 & 0 & 89 \\
\hline
	\end{tabular}
	\endgroup
	\caption{$n = 100, p = 5$} 
\end{table}

\begin{table}[H]
	\centering
	\begingroup\footnotesize
	\begin{tabular}{lrrrrrrrrr}
		\hline
		& $H_{n}$ & $D_{n}$ & $T_{n}$ & EN & HV & HZ & HJG & MS & MK \\ 
		\hline
		$N(0, I)$ & 5 & 5 & 5 & 4 & 5 & 5 & 5 & 5 & 5 \\ 
		$N(0, S)$ & 5 & 5 & 5 & 5 & 5 & 5 & 5 & 5 & 5 \\ 
		$N(m, I)$ & 5 & 5 & 5 & 5 & 5 & 5 & 5 & 5 & 5 \\ 
		$N(m, S)$ & 5 & 5 & 5 & 5 & 5 & 5 & 5 & 5 & 5 \\ 
		$U(0, 1)^p$ & 0 & 100 & 100 & 9 & 0 & 76 & 0 & 0 & 100 \\ 
		Beta$(0.5, 0.5)^p$ & 0 & 100 & 100 & 46 & 0 & 100 & 0 & 0 & 100 \\ 
		Beta$(2, 2)^p$ & 0 & 100 & 100 & 2 & 0 & 25 & 0 & 0 & 87 \\
\hline
	\end{tabular}
	\endgroup
	\caption{$n = 100, p = 10$} 
	\label{table:n100p10}
\end{table}

%
%
%
%
%
%
%

\bibliographystyle{Chicago}
\bibliography{normal_test_reference}

\end{document}